\documentclass{amsart}
\usepackage{amsmath}
\usepackage{amssymb}
\usepackage{amsthm}
\usepackage{verbatim}

\newtheorem{theorem}{Theorem}[section]
\newtheorem{lemma}{Lemma}

\newtheorem{cor}{Corollary}
\newtheorem{corollary}{Corollary}
\newtheorem{prop}{Proposition}
\newtheorem{remark}{Remark}

\everymath{\displaystyle}

\begin{document}

\title{On Rotating Star Solutions to the Non-isentropic Euler-Poisson Equations}
\author{Yilun Wu}
\address{Department of Mathematics, University of Michigan, 
Ann Arbor, MI 48109}
\email{yilunwu@umich.edu}

\begin{abstract}
This paper investigates rotating star solutions to the Euler-Poisson equations with a non-isentropic equation of state. As a first step, the equation for gas density with a prescribed entropy and angular velocity distribution is studied. The resulting elliptic equation is solved either by the method of sub and supersolutions or by a variational method, depending on the value of the adiabatic index. The reverse problem of determining angular velocity from gas density is also considered.
\end{abstract}

\maketitle
\section{Introduction}

A Newtonian star is modeled by a body of fluids satisfying the Euler-Poisson equations in three spatial dimensions:

\begin{equation}\label{eq:Euler-Poisson system}
\begin{cases}
\rho_t + \nabla\cdot(\rho\mathbf{v}) = 0 \\
(\rho\mathbf{v})_t + \nabla\cdot(\rho\mathbf{v}\otimes\mathbf{v}) + \nabla p = -\rho\nabla\phi\\
\end{cases}
\end{equation}
Here $\rho$, $\mathbf{v}$, $p$ are the density, velocity, and pressure of the fluids. $\phi$ is the Newtonian potential of $\rho$ given by 
\begin{equation}\label{Newtonian potential}
\phi(\mathbf{x})=-G\int_{\mathbb{R}^3}\frac{\rho(\mathbf{y})}{|\mathbf{x}-\mathbf{y}|}d\mathbf{y}.
\end{equation}
Here $G$ is Newton's universal constant of gravitation.

There is a long history of investigation of equilibrium solutions to \eqref{eq:Euler-Poisson system}. 
Non-rotating stars can be interpreted as spherically symmetric stationary solutions to \eqref{eq:Euler-Poisson system}. 
The theory is described in detail in Chandrasekhar's classical work \cite{chandrasekhar1939introduction}. 
Rotating stars, on the other hand, are axisymmetric equilibrium solutions to \eqref{eq:Euler-Poisson system}, whose velocity field contains only the azimuthal component. In cylindrical coordinates $(r,\theta, z)$, the assumptions on rotating star solutions can be formulated as setting $\rho$, $p$, $\mathbf{v}$ to be functions of $r$ and $z$ only, and $\mathbf{v}$ to contain only the $\mathbf{e}_{\theta}$ component. Here $\mathbf{e}_r$, $\mathbf{e}_{\theta}$, $\mathbf{e}_z$ are the unit vectors in cylindrical coordinates. With these assumptions, the mass conservation equation in \eqref{eq:Euler-Poisson system} is identically satisfied. Writing $\mathbf{v}(r,z)=r\Omega(r,z)\mathbf{e}_{\theta}$, the momentum conservation equation in \eqref{eq:Euler-Poisson system} becomes
\begin{equation}\label{eq: rotating star eq}
\begin{cases}
p_r=\rho (-\phi)_r + \rho r\Omega^2 \\
p_z=\rho (-\phi)_z
\end{cases}
\end{equation}
or 
\begin{equation}\label{eq: rotating star, vector}
\frac{\nabla p}{\rho}= - \nabla \phi + r\Omega ^2 \mathbf{e}_r.
\end{equation}
Auchmuty and Beals \cite{auchmuty1971variational} initiated the search for solutions to \eqref{eq: rotating star eq}. They prescribed the angular velocity distribution and the total mass, and found solutions for the density function $\rho$. In order to close the underdetermined system \eqref{eq:Euler-Poisson system}, they assumed an isentropic equation of state:
\begin{equation}\label{isentropic equation of state}
p=g(\rho)
\end{equation}
for some given function $g$ (Note that condition \eqref{isentropic equation of state} is more often called barotropic equation of state in physics literature). As $\phi$ is given as \eqref{Newtonian potential}, and $p$ as \eqref{isentropic equation of state}, \eqref{eq: rotating star eq} consists of two equations for one unknown function $\rho$, hence appears to be overdetermined. To understand how this problem is well posed, let us take the curl of \eqref{eq: rotating star, vector}:

\begin{equation}
\nabla \times \bigg( \frac{\nabla p}{\rho} \bigg) = \nabla \times (r\Omega ^2 \mathbf{e}_r),
\end{equation}
which then simplifies to 
\begin{equation}\label{curl of Euler Poisson}
\frac{\nabla p \times \nabla \rho}{\rho ^2} = r\frac{\partial \Omega ^2}{\partial z}\mathbf{e}_{\theta}.
\end{equation}
\eqref{curl of Euler Poisson} implies the following important proposition:

\begin{prop}\label{prop: zero curl condition}
Both sides of \eqref{eq: rotating star, vector} are curl free if and only if $\nabla p \times \nabla \rho =0$ if and only if $\Omega^2$ depends only on $r$.
\end{prop}

A useful corollary of this proposition is as follows.

\begin{cor}\label{cor: 1}
If $p=g(\rho )$ is a given function of $\rho$, $\Omega ^2$ can only depend on $r$. 
\end{cor}

When $\Omega^2$ is prescribed as a function of $r$, and $p =g(\rho)$ for some given $g$, \eqref{eq: rotating star, vector} is curl free. In fact it can be written as a gradient.

\begin{equation}\label{eq: gradient Bernoulli}
\nabla \big( A(\rho) \big)= - \nabla \phi + \nabla J,
\end{equation}
where 
\begin{equation}
A(s) = \int_0^s \frac{g'(t)}{t}dt 
\end{equation}
and
\begin{equation}
J(r) = \int_0^r s\Omega^2(s) ds.
\end{equation}
From \eqref{eq: gradient Bernoulli}, we get
\begin{equation}\label{Bernoulli relation}
A(\rho)= -\phi + J(r) + C 
\end{equation}
for some constant $C$. With $\Omega$, $g$ prescribed, and $\phi$ given as \eqref{Newtonian potential}, \eqref{Bernoulli relation} is a single equation for $\rho$. It is this equation that Auchmuty and Beals studied with a variational method.

Now let us try to generalize and allow $\Omega^2$ to depend on $r$ and $z$. Corollary \ref{cor: 1} will force us to use a non-isentropic equation of state:
\begin{equation}\label{eq: non-isentropic eq of state}
p=g(\rho, s)
\end{equation}
where $s$ is entropy. The full Euler-Poisson system has another equation for energy conservation, which we have been ignoring until now:
\begin{equation}\label{eq: energy conservation}
\bigg(\frac{1}{2}\rho |\mathbf{v}|^2 + \rho e\bigg)_t +\nabla\cdot \bigg(\bigg(\frac{1}{2}\rho |\mathbf{v}|^2 + \rho e\bigg)\mathbf{v}\bigg)=-\rho\nabla\phi\cdot\mathbf{v}-\nabla\cdot(p\mathbf{v}).
\end{equation}
Here $e$ is specific internal energy. By the second law of thermodynamics,
\begin{equation}\label{eq: second law of thermodynamics}
de=T(\rho ,s)ds + \frac{p(\rho,s)}{\rho^2}d\rho,
\end{equation}
therefore
\begin{equation}
e(\rho,s)=\int_0^{\rho} \frac{g(\xi, s)}{\xi ^2}d\xi.
\end{equation}
Here we have assumed the condition
$$e(0,s)=0.$$
After a simple calculation, we get
\begin{align}\label{eq: entropy transport}
s_t+\mathbf{v}\cdot\nabla s&=0.
\end{align}
We join \eqref{eq: entropy transport} with \eqref{eq:Euler-Poisson system} to get the non-isentropic Euler-Poisson system
\begin{equation}\label{eq: non-isentropic Euler-Poisson}
\begin{cases}
\rho_t + \nabla\cdot(\rho\mathbf{v}) = 0 \\
(\rho\mathbf{v})_t + \nabla\cdot(\rho\mathbf{v}\otimes\mathbf{v}) + \nabla p = -\rho\nabla\phi\\
s_t+\mathbf{v}\cdot\nabla s=0
\end{cases}
\end{equation}
with Newtonian potential given as \eqref{Newtonian potential} and equation of state \eqref{eq: non-isentropic eq of state}.

The study of rotating star solutions to \eqref{eq: non-isentropic Euler-Poisson} presents a new problem and demands further investigation. As before, we assume $\rho$, $p$, $\mathbf{v}$, $s$ to depend only on $r$ and $z$, and $\mathbf{v}$ to contain only the $\mathbf{e}_{\theta}$ component. Under these assumptions, the mass conservation equation and entropy transport equation in \eqref{eq: non-isentropic Euler-Poisson} are identically satisfied, therefore we again arrive at equations \eqref{eq: rotating star eq} or \eqref{eq: rotating star, vector}. The difference with the classical Auchmuty and Beals case is that, by proposition \ref{prop: zero curl condition}, neither sides of \eqref{eq: rotating star, vector} has vanishing curl in general, and hence \eqref{eq: rotating star, vector} is no longer a gradient. It is not obvious how one can recast \eqref{eq: rotating star, vector} as the Euler-Lagrange equation of some energy functional as in the classical Auchmuty and Beals case.

With the angular velocity distribution prescribed and $\phi$ given as \eqref{Newtonian potential}, system \eqref{eq: rotating star eq} is a set of two equations for two unknowns $\rho$ and $p$. Different from the isentropic case, the introduction of the equation of state \eqref{eq: non-isentropic eq of state} does not decrease the number of unknowns, but merely transforms the unknowns to $\rho$ and $s$. For definiteness, from now on we will use the equation of state:
\begin{equation}\label{eq: eq of state of ideal gas}
p = g(\rho ,s) = e^s \rho ^{\gamma}.
\end{equation}
Here $\gamma$ is a constant, and $e=\exp (1)$ is base of natural logarithm. Notice that we have employed the same letter for specific internal energy, but the confusion should be minimal by watching the contexts they appear in. This equation of state is very general, it for instance incorporates arbitrary thermal processes of an ideal gas (see Courant and Friedrichs \cite{courant1976supersonic}). Inserting \eqref{eq: eq of state of ideal gas} into \eqref{eq: rotating star eq} or \eqref{eq: rotating star, vector}, we get
\begin{equation}\label{eq: rotating star eq - entropy}
\begin{cases}
(e^s \rho ^{\gamma})_r=\rho (-\phi)_r + \rho r\Omega^2 \\
(e^s \rho ^{\gamma})_z=\rho (-\phi)_z
\end{cases}
\end{equation}
or equivalently,
\begin{equation}\label{eq: rotating star eq - vector - entropy}
\frac{\nabla (e^s \rho ^{\gamma})}{\rho}= - \nabla \phi + r\Omega ^2 \mathbf{e}_r.
\end{equation}
This is a system of two equations for two unknowns $(\rho, s)$. The search for solutions to \eqref{eq: rotating star eq - entropy} with prescribed angular velocity distribution is still an open problem. 

As a first step to study these questions, we will treat two problems in this paper. One is to take the divergence of \eqref{eq: rotating star eq - vector - entropy} and study solutions to
\begin{equation}\label{eq: divergence of rotating star eq}
\nabla \cdot \bigg(\frac{1}{\rho}\nabla (e^s \rho^{\gamma})\bigg) = - 4\pi G \rho + \nabla \cdot (r\Omega ^2 \mathbf{e}_r) 
\end{equation}
with prescribed entropy. The other is to take the converse path, and consider what conditions on $\rho$ can one impose to solve \eqref{eq: rotating star eq - entropy} for some entropy $s$, and nonnegative angular velocity field $\Omega^2$.

\section{Statement of Results}

Let us consider the following axisymmetric equilibrium non-isentropic Euler-Poisson equation in $\mathbb{R}^3$:
\begin{equation}\label{chap4: eq: non-isentropic vector rotating star}
\frac{\nabla p}{\rho} = B\rho + r\Omega^2 \mathbf{e}_r
\end{equation}
with equation of state
\begin{equation}
p = e^s \rho ^{\gamma}.
\end{equation}
Here $\gamma$ is a constant called the adiabatic index. The divergence of \eqref{chap4: eq: non-isentropic vector rotating star} is 
\begin{equation}\label{chap4: intro: eq: divergence of rotating star eq}
\nabla \cdot \bigg(\frac{\nabla (e^s \rho^{\gamma})}{\rho}\bigg) = - 4\pi \rho + \nabla \cdot (r\Omega ^2 \mathbf{e}_r).
\end{equation}
After the change of variable
\begin{equation}\label{chap4: intro: eq: change of variables}
w=\frac{\gamma}{\gamma-1}e^{\frac{\gamma-1}{\gamma }s}\rho^{\gamma -1},
\end{equation}
\eqref{chap4: intro: eq: divergence of rotating star eq} becomes
\begin{equation}\label{chap4: intro: eq: w equation}
\nabla \cdot (e^{\alpha s}\nabla w) + Ke^{-\alpha s}|w|^q-f=0,
\end{equation}
where
\begin{equation}
q  =\frac{1}{\gamma -1},\quad \alpha =\frac{1}{\gamma}, \quad K=4\pi \bigg(\frac{\gamma-1}{\gamma}\bigg)^{\frac{1}{\gamma-1}}, 
\end{equation}
and 
\begin{equation}\label{chap4: eq: defining f}
f= 2\Omega^2 + r\frac{\partial \Omega^2}{\partial r} = 2\Omega\frac{\partial}{\partial r}(r\Omega).
\end{equation}

With $s$ and $\Omega$ prescribed, Luo and Smoller \cite{luo2004rotating} considered \eqref{chap4: intro: eq: w equation} and obtained some existence results when the entropy is assumed to be either constant or radially dependent, and a non-existence result when the entropy is non-constant. In the present paper we will find some existence results for \eqref{chap4: intro: eq: w equation} with axisymmetric entropy. Standard elliptic theory (Gilbarg and Trudinger \cite{gilbarg2001elliptic}) can solve the Dirichlet problem to \eqref{chap4: intro: eq: w equation} on bounded domains given suitable range of $q$, but in order to conclude positivity of $w$ inside the domain, it is desirable that $f$ be negative. Unfortunately for most physically interesting $\Omega$, $f$ is positive. For example, constant $\Omega$ will produce a positive $f$. Therefore the gist of the argument is to show existence of positive solutions.

We have the following
\begin{theorem}\label{chap4: solution 0<q<1}
Let $f$ and $s$ be given axisymmetric smooth functions. If $0<q<1$ $(\gamma >2)$, and $\mathbf{x} \cdot \nabla s \leq 0$, then there is a finite ball centered at the origin on which there exists an axisymmetric positive smooth solution to \eqref{chap4: intro: eq: w equation} with zero boundary value. 
\end{theorem}
The condition on entropy has the physical interpretation that it is decreasing in the radial direction, so that the star is thermally more active the further one goes down surface.

The $q>1$ case is more difficult. Let us take a look at a simple ODE model $u''+\lambda u^q=0$. Suppose $q>1$. In order for $u$ to stay positive, symmetric about the origin, and be zero on the boundary of a given symmetric domain, $u(0)$ will be unbounded as $\lambda$ gets close to $0$. Therefore there is no a priori bound for Leray-Schauder type arguments. However, if one is allowed to rescale the velocity field, the equation can still be solved. The results are as follows:

\begin{theorem}\label{chap4: solution 1<q<3 bounded}
Let $f$ and $s$ be given axisymmetric smooth bounded functions. Suppose there exists a $c>0$ such that $f \geq c$, and suppose $1<q<3$ $\bigg( \frac{4}{3}<\gamma <2 \bigg)$. Then for any $R>0$, and sufficiently large $P>0$, there exists a non-negative axisymmetric function $w$ in $H_0^1(S_R)$, and a $\lambda >0$, such that 
$w$ is smooth on its own positive set and satisfies
\begin{equation}\label{chap4: eq: rescaled eq}
\nabla \cdot (e^{\alpha s}\nabla w) + Ke^{-\alpha s}w^q-\lambda f=0,
\end{equation}
and
\begin{equation}
\int_{S_R} f w ~d\mathbf{x} = P.
\end{equation}
\end{theorem}

Here $S_R$ is a ball of radius $R$ centered at the origin. The positive set of $w$ will turn out to be open, so there is no ambiguity in defining \eqref{chap4: eq: rescaled eq}. Since $\lambda>0$, \eqref{chap4: intro: eq: w equation} is solved with a rescaled the velocity field (compare \eqref{chap4: eq: defining f}). Also
\begin{align*}
P&=\int_{S_R}f w ~d\mathbf{x}\\
&\leq \int_{S_R} Cw ~d\mathbf{x}\\
&\leq \tilde{C}\int_{S_R} \rho ^{\gamma -1} ~d\mathbf{x}\\
&\leq \tilde{\tilde{C}}\bigg(\int_{S_R} \rho ~d\mathbf{x}\bigg)^{\gamma -1}. \\
\end{align*}
Therefore the largeness of $P$ implies the largeness of the total mass in this case. 

The method for deriving this result is variational. It is possible to extend the variational method to allow functions defined on the entire $\mathbb{R}^3$ once we find a way to address the lost of compactness issue.

\begin{theorem}\label{chap4: solution 1<q<3 unbounded}
Let $f$ and $s$ be given axisymmetric smooth functions. Suppose $s$ is bounded, $f \geq c>0$, and $1<q<3$ $\bigg( \frac{4}{3}<\gamma <2\bigg)$. Then for sufficiently large $P>0$, there exists a non-negative axisymmetric function $w$ in $H^1(\mathbb{R}^3)$, and a $\lambda >0$, such that $w$ is smooth on its own positive set, and satisfies \eqref{chap4: eq: rescaled eq} and
\begin{equation}
\int_{\mathbb{R}^3} f w ~d\mathbf{x} = P.
\end{equation}
\end{theorem}

Another way of investigating solutions to \eqref{chap4: eq: non-isentropic vector rotating star} is by prescribing $\rho$ and solving for $p$ and $\Omega ^2$. Apart from being suitably smooth, an obvious requirement for $p$ and $\Omega ^2$ is that they should be positive where $\rho$ is positive. Furthermore, $p$ should be zero on the boundary of the positive set of $\rho$. It is possible to develop conditions on $\rho$ that will guarantee the existence of such $p$ and $\Omega ^2$. To find out what conditions on $\rho$ are natural, we observe some features of the classical Auchmuty and Beals solutions with isentropic equation of state. In \cite{caffarelli1980shape}, Caffarelli and Friedman studied the shapes of the Auchmuty and Beals solutions. Some of their results can be summarized as follows:

\begin{prop}\label{chap4: prop: shape of A-B solutions}
Assume $\Omega^2$ is analytic, and the equation of state is given by
\begin{equation}
p=c \rho ^{\gamma}
\end{equation}
for some $\frac{4}{3} < \gamma < 2$ ($ 1<q<3$). Then the Auchmuty and Beals solution $\rho$ to \eqref{chap4: eq: non-isentropic vector rotating star} has the following properties:

\begin{enumerate}
\item{
Let $D=\big\{\mathbf{x}\in \mathbb{R}^3 ~\big|~ \rho(\mathbf{x})>0 \big\}$, then $\bar{D}$ is compact, $\partial D$ is smooth and $D$ is a finite union of sets of the form $\big\{(r,z) ~\big|~ 0\leq a<r<b, |z|< \psi(r)\}$, where $\psi$ is a function vanishing at the end points except if $a=0$. $\rho \in C^{0,\beta}(\mathbb{R}^3)\cap C^{\infty}(D)$ for some $\beta>0$.
}
\item{
$\rho(r,z)=\rho(r,-z)$.
}

\item{
$\rho_z(r,z)>0$ for $(r,z) \in D$, $r>0$, and $z<0$.
} 

\item{
$\rho_{zz}(r,0)<0$ for $(r,0) \in D$.
}
\end{enumerate}
\end{prop}

Motivated by this result, we will prove the following
\begin{theorem}\label{chap4: thm: solve velocity field: smooth}
Let $\rho$ be an axisymmetric nonnegative function such that
\begin{enumerate}
\item{
$\rho\in C^{k}(\bar{D})$ ($k\geq 2$), where $D$ is a finite union of sets of the form $\big\{(r,z) ~\big|~ 0\leq a<r<b, |z|< \psi(r)\}$, where $\psi$ is a function vanishing at the end points except if $a=0$. Also assume $\partial D$ is smooth, $\rho>0$ on $D$, $\rho=0$ on $\partial{D}$.
}
\item{
$\rho(r,z)=\rho(r,-z)$.
}
\item{
$\rho_r(B\rho)_z-\rho_z(B\rho)_r\geq 0$ for $z<0$.
}
\item{
$\rho_z>0$ for $z<0$.
}

\end{enumerate}
Also assume the following is satisfied:

\indent (a) There is a $c>0$, such that $\rho_{zz}<-c$ on $\{z=0\}\cap \partial D$.

Then \eqref{chap4: eq: non-isentropic vector rotating star} is solvable for a nonnegative angular velocity function $\Omega^2 \in C^{k-2}(D)\cap C^0(\bar{D})$ and a positive pressure $p\in C^{k}(\bar{D})$, such that $p=0$ on $\partial D$.
\end{theorem}

\begin{remark}
If $\nabla \rho$ and $\nabla (B\rho)$ point approximately to the center of the star, condition 3 in theorem \ref{chap4: thm: solve velocity field: smooth} means that the gradient of $\rho$ is more inclined with respect to the plane $\{z=0\}$ than the gravity force. Simple calculations with ellipsoids suggest that shapes that are wider at the equator tend to satisfy condition 3.
\end{remark}

\begin{remark}
In general the solution obtained from Theorem \ref{chap4: thm: solve velocity field: smooth} has an angular velocity distribution that depends both on $r$ and $z$. By Proposition \ref{prop: zero curl condition}, such solutions will require a non-isentropic equation of state.
\end{remark}

It is desirable to relax the regularity conditions of $\rho$ at the boundary, since for some $\gamma$, the Auchmuty and Beals solutions are only H\"{o}lder continuous at the boundary. A similar result with weaker boundary regularity needs more control on the derivatives when close to the boundary. Here is one way of formulating the conditions:

\begin{theorem}\label{chap4: thm: solve velocity field holder rho}
Let $\rho$ be an axisymmetric nonnegative function such that
\begin{enumerate}
\item{
$\rho\in C^{2}(D)\cap C^{0,\beta}(\bar{D})$, for some $0<\beta<1$, where $D$ is a finite union of sets of the form $\big\{(r,z) ~\big|~ 0\leq a<r<b, |z|< \psi(r)\}$, where $\psi$ is a function vanishing at the end points except if $a=0$. Also assume that $\partial D$ is smooth, convex at $(0,\pm \psi(0))\in \partial D$ (if there are such points), i.e., the interior of the segment $(0, \psi(0)) - (r,\psi(r))$ lies in $D$ for $r$ sufficiently small. $\rho>0$ on $D$, $\rho=0$ on $\partial{D}$.
}
\item{
$\rho(r,z)=\rho(r,-z)$.
}
\item{
$\rho_r(B\rho)_z-\rho_z(B\rho)_r\geq 0$ for $z<0$.
}
\item{
$\rho_z>0$ for $z<0$.
}

\item{
$\forall \epsilon>0$, $\exists C>0$ such that on $D\cap \{|z|\geq \epsilon\}$: $|\rho_r|\leq C|\rho_z|$, $|\rho_{rr}|\leq C|\rho_z|$, $|\rho_{rz}|\leq C |\rho_z|$.
} 
\end{enumerate}
Also assume that one of the following is satisfied:

\indent (a) $\displaystyle \frac{\rho_{rz}}{\rho_{zz}}$ and $\displaystyle \frac{\rho_r}{\rho_{zz}}$ are bounded in a neighbourhood of $\{z=0\}\cap \partial D$. 

\indent (a') $\rho_r\leq 0$ in a neighbourhood of $\{z=0\}\cap \partial D$.

\indent (a'') $\displaystyle \frac{z \rho_r}{\rho_z}$ is bounded on $U\setminus \{z=0\}$, where $U$ is some neighbourhood of $\{z=0\}\cap \partial D$. 

Then \eqref{chap4: eq: non-isentropic vector rotating star} is solvable for a nonnegative angular velocity function $\Omega^2 \in C^0(D)\cap L^{\infty}(D)$ and a positive pressure $p\in C^1(D)\cap C^0({\bar{D}})$, such that $p=0$ on $\partial D$.
\end{theorem}

\begin{remark}
If $D$ has only one connected component containing the origin, $\{z=0\}\cap \partial D$ is the equator, and since $\rho$ is zero on the equator and positive in the interior of $D$, the condition (a') is most likely satisfied in this case. Condition (a'') is equivalently to $\frac{z}{r} \bigg/ \frac{\rho_z}{\rho_r}$ being bounded on $U\setminus \{z=0\}$ and has the geometrical interpretation that the when $\mathbf{x}$ gets close to $\{z=0\}\cap \partial D$, the inclination of $\mathbf{x}$ to the horizontal plane is bounded by the inclination of $\nabla \rho(\mathbf{x})$.
\end{remark}

\section{Existence of Solution for High Adiabatic Index}
\label{chap4: high index existence}

Without loss of generality, we may absorb $\alpha$ into $s$ in \eqref{chap4: intro: eq: w equation} and work with
\begin{equation}\label{chap4: eq: main eq 1}
\nabla \cdot (e^{s}\nabla w) + Ke^{-s}w^q-f=0.
\end{equation}
We first find a subsolution to this equation.

\begin{lemma}\label{chap4: subsolution}
If $\mathbf{x} \cdot \nabla s \leq 0$, there is a ball of radius $R$, denoted by $S_R$, centered at the origin, on which there is a smooth spherically symmetric positive function $\underline{u}$ with zero boundary value satisfying
\begin{equation}\label{chap4: subsolution inequality}
\nabla \cdot (e^{s}\nabla \underline{u}) + Ke^{-s}\underline{u}^q-f \geq 0.
\end{equation}
\end{lemma}

\begin{proof}
Let $A_1$, $A_2$ be two positive constants such that $Ke^{-2s}\geq A_1$, $e^{-s}f \leq A_2$. We look for a positive function $\underline{u}$ on a ball which satisfies
\begin{equation}\label{chap4: inequality 1}
\Delta \underline{u} + A_1\underline{u}^q - A_2 \geq 0.
\end{equation}
By lemma 3.1 in \cite{smoller1989existence}, we only need to check that the primitive of $g(t)=A_1t^q- A_2$, which is $G(t)=\frac{A_1}{q+1}t^{q+1}-A_2t$, satisfies $G(t)>0$ for some $t>0$. But this is certainly true for large enough $t$. It follows that there is a ball of radius $R$, and a spherically symmetric positive solution $\underline{u}$ of \eqref{chap4: inequality 1} on this ball with zero boundary value, which satisfies $\mathbf{x} \cdot \nabla \underline{u}<0$. By the definition of $A_1$ and $A_2$, we have
\begin{equation}
\Delta \underline{u} + Ke^{-2s}\underline{u}^q-e^{-s}f \geq \Delta \underline{u} + A_1 \underline{u}^q-A_2 \geq 0.
\end{equation}
Furthermore, by 
\begin{equation}
\nabla \underline{u} = -\frac{|\nabla \underline{u}|}{|\mathbf{x}|}\mathbf{x},
\end{equation}
we have
\begin{equation}
\nabla s \cdot \nabla \underline{u} = - (\mathbf{x} \cdot \nabla s) \frac{|\nabla \underline{u}|}{|\mathbf{x}|} \geq 0.
\end{equation}
Therefore,
\begin{equation}
\Delta \underline{u} +\nabla s \cdot \nabla \underline{u} +Ke^{-2s}\underline{u}^q -e^{-s}f \geq 0,
\end{equation}
which differs from \eqref{chap4: subsolution inequality} only by a factor of $e^{-s}$. Hence the assertion is proved.
\end{proof}

Having found a subsolution to \eqref{chap4: eq: main eq 1}, we now only need a supersolution to produce a genuine solution. That is given by 

\begin{lemma}\label{chap4: supersolution}
Suppose $0<q<1$. There is a smooth positive function $\bar{u}$ on $\overline{S_R}$, such that $\bar{u}\geq \underline{u}$ on $\overline{S_R}$, and satisfies 
\begin{equation}
\nabla \cdot (e^{s}\nabla \bar{u}) + Ke^{-s}\bar{u}^q-f \leq 0.
\end{equation} 
\end{lemma}

\begin{proof}
Let $C=\| \underline{u} \|_{L^{\infty}(S_R)}$, and $M=\| f \|_{L^{\infty}(S_R)}$. Let $g(t)\geq 0$ be a smooth function on $\mathbb{R}$ such that
\begin{equation}
g(t)=
\begin{cases}
t^q & \mbox{if } t \geq C\\
0 & \mbox{if } t \leq 0
\end{cases}
\end{equation}
and $0\leq g'(t) \leq 2C^{q-1}$ when $0<t<C$. We look for a solution to the equation:
\begin{equation}\label{chap4: weakened equation}
\nabla \cdot (e^{s}\nabla u) + Ke^{-s}g(u+C) + M = 0
\end{equation}
by the standard Leray-Schauder estimate. For that we define
\begin{align*}
A :  H^1_0(S_R) & \rightarrow H^1_0(S_R)\\
u & \mapsto v
\end{align*}
by 
\begin{align*}
\nabla \cdot (e^{s}\nabla v) + Ke^{-s}g(u+C) + M = 0 \ \ \ \ &\mbox{on } S_R \\
v=0 \ \ \ \ &\mbox{on } \partial S_R
\end{align*}
By the definition of $g(t)$ we have
\begin{align*}
\big( g(u+C) \big) ^2 & \leq C^{2q} + (u+C)^{2q},\\
& \leq \tilde{C}(1+ u^2)
\end{align*}
where $\tilde{C}$ is a constant which will be enlarged appropriately in the following. Therefore $A(u)\in H^2(S_R)$, and 
\begin{equation}
\| A(u) \|_{H^2(S_R)}\leq \tilde{C}( 1+ \| u \|_{H^1_0(S_R)}).
\end{equation}
It follows easily that $A$ is continuous and compact. Furthermore if $u=t A(u)$, for $0\leq t \leq 1$, we have
\begin{equation}
\nabla \cdot (e^s \nabla u) + t (Ke^{-s}g(u+C) +M) = 0
\end{equation}
weakly. Therefore for some $c>0$
\begin{align*}
& c\int_{S_R}| \nabla u |^2\\
\leq & \int_{S_R}e^s |\nabla u |^2\\
= & t \int _{S_R}Ke^{-s}g(u+C)u + Mu.\\
\end{align*}
Notice that $g(u+C)\leq C^q + \tilde{C}(u^q+C^q)$,
\begin{align*}
& c\int_{S_R}| \nabla u |^2\\
\leq & \tilde{C}(1+\int_{S_R}u^{q+1}+\tilde{C}Mu)\\
\leq & \tilde{C}(1+C(\epsilon) + \epsilon \int_{S_R}u^2)\\
\leq & \tilde{C}(C(\epsilon) + \epsilon \|u\|_{H^1_0(S_R)}^2).
\end{align*}
Here the constants $\tilde{C}$ and $C(\epsilon)$ are enlarged appropriately from line to line. Let us now choose $\epsilon$ so small that $\tilde{C}\epsilon < \frac{c}{2}$. It follows that $\big\{ u ~\big|~ u=t A(u), 0\leq t \leq 1  \big\}$ is bounded in $H^1_0(S_R)$. Therefore there exists a $u$ in $H^1_0(S_R)$ solving \eqref{chap4: weakened equation}. By the Sobolev imbedding theorem, $u\in H^2(S_R) \subset W^{1,6}(S_R) \subset C^{0,\frac{1}{2}}(\overline{S_R})$. Since
\begin{align*}
& |g(u(x)+C)-g(u(y)+C)| \\
& \leq |g'(\theta)||u(x)-u(y)|\\
& \leq \max (2C^{q-1},q(C+\|u\|_{C^0(\overline{S_R})})^{q-1})[u]_{0,\frac{1}{2};S_R}|x-y|^{\frac{1}{2}},
\end{align*}
where $\theta$ is between $u(x)+C$ and $u(y)+C$, it follows that $g(u+C)\in C^{0,\frac{1}{2}}(\overline{S_R})$. Elliptic regularity estimates imply $u\in C^{2,\frac{1}{2}}(\overline{S_R})$, and an iteration of the regularity estimates imply that $u$ is smooth. Now by the classical maximum principle, $u\geq 0$ on $S_R$, therefore $u$ solves
\begin{equation}
\nabla \cdot (e^{s}\nabla u) + Ke^{-s}(u+C)^q + M = 0.
\end{equation}
Hence
\begin{equation}
\nabla \cdot (e^{s}\nabla (u+C)) + Ke^{-s}(u+C)^q - f \leq 0.
\end{equation}
Letting $\bar{u}=u+C$ completes the proof.
\end{proof}

\begin{proof}[Proof of theorem \ref{chap4: solution 0<q<1}]
It follows from lemma \ref{chap4: subsolution}, lemma \ref{chap4: supersolution}, and a standard construction (see Smoller \cite{smoller1983shock}) that a solution to \eqref{chap4: eq: main eq 1} exists. The construction also guarantees the resulting solution to be axisymmetric if $\underline{u}$ is.
\end{proof}

\section{Variational Formulation}
\label{chap4: variation}

The main purpose of this section is to show existence of minimizer of the following energy functional:
\begin{equation}\label{chap4: energy functional}
E(w)=\int_{\mathbb{R}^3}\bigg(\frac{e^{ s}}{2}|\nabla w|^2-\frac{K}{q+1}w^{q+1}e^{- s}\bigg) ~d\mathbf{x}
\end{equation}
subject to the constraint:
\begin{equation}\label{chap4: constraint}
N(w)=\int_{\mathbb{R}^3}f w ~d\mathbf{x} =P.
\end{equation}
where $f$ is assumed to be locally bounded, and
\begin{equation}\label{chap4: f bound}
f\geq c >0.
\end{equation}
We take the set $W_P$ of admissible functions to be 
\begin{align}\label{chap4: eq: admissible class}
H^1(\mathbb{R}^3)\cap L^1(\mathbb{R}^3) \cap \big\{w: \mathbb{R}^3\to \mathbb{R}, ~w\geq 0\ \text{a.e.}, w\ \text{is axisymmetric}, N(w)=P\big\}.
\end{align}
In fact, one has
\begin{prop}\label{chap4: thm: existence of minimizer, infinite}
If $1< q<3$, there exists a minimizer in $W_P$ of the energy functional $E$ for $P$ sufficiently large.
\end{prop}

We will apply this proposition to construct solutions to \eqref{chap4: eq: non-isentropic vector rotating star} when $1<q<3$ and the domain is infinite.

The proof will need a bound of the $L^{q+1}$ norm by the $L^p$ norm and the $L^2$ norm of the derivative. We will only concern ourselves with the case in $\mathbb{R}^3$. This is given by the following inequality (see, for example, \cite{nirenberg2011elliptic}).
\begin{prop}[Gagliaro-Nirenberg inequality]
Let $1\leq p < 6$, $p\leq q+1 \leq 6$. If $w\in L^p(\mathbb{R}^3) \cap H^1(\mathbb{R}^3)$, then $\exists C>0$, such that
\begin{equation*}
\|w\|_{L^{q+1}(\mathbb{R}^3)}\leq C\|\nabla w\|_{L^2(\mathbb{R}^3)}^a \|w\|_{L^p(\mathbb{R}^3)}^{1-a}.
\end{equation*} 
If $w\in L^p(\mathbb{R}^3\setminus S_R) \cap H^1(\mathbb{R}^3\setminus S_R)$, where $S_R$ is the ball centered at the origin with radius $R>R_0>0$, then $\exists C(R_0)>0$, such that
\begin{equation*}
\|w\|_{L^{q+1}(\mathbb{R}^3\setminus S_R)}\leq C\|\nabla w\|_{L^2(\mathbb{R}^3\setminus S_R)}^a \|w\|_{L^p(\mathbb{R}^3\setminus S_R)}^{1-a}.
\end{equation*} 
In both of these inequalities, 
\begin{equation*}
a=\frac{\frac{1}{p}-\frac{1}{q+1}}{\frac{1}{p}-\frac{1}{6}}.
\end{equation*}
\end{prop}

Notice when $q\leq 5$, $0<a\leq 1$. This is the useful range of exponents for us. With the Gagliardo-Nirenberg inequality, we can show that $E$ is bounded from below on $W_P$. 

\begin{lemma}\label{chap4: bounded from below lemma}
Suppose $w\in L^1(\mathbb{R}^3) \cap H^1(\mathbb{R}^3)$, $N(w)=P$, and $q<3$, then there exists a constant $C$ depending only on $P$, such that 
\begin{equation*}
E(w)\geq \frac{1}{2}\int_{\mathbb{R}^3}\frac{e^{ s}}{2}|\nabla w|^2~d\mathbf{x} -C.
\end{equation*}
\end{lemma}

\begin{proof}
Since $s$ is bounded,
\begin{equation*}
\int w^{q+1}e^{-s} ~d\mathbf{x} \leq C\int w^{q+1} ~d\mathbf{x} = C\|w\|_{q+1}^{q+1}.
\end{equation*}
By the Gagliardo-Nirenberg inequality, we have 
\begin{align*}
& C\|w\|_{q+1}^{q+1}\\
\leq &C\|\nabla w\|_{L^2}^{a(q+1)}\|w\|_{L^1}^{(1-a)(q+1)}\\
\leq &C(P)\|\nabla w \|_{L^2}^{a(q+1)}.
\end{align*}
The last inequality follows from the boundedness of $s$, \eqref{chap4: constraint}, and \eqref{chap4: f bound}.

Since $q<3$, an easy calculation shows $a(q+1)<2$. By an elementary inequality we have
\begin{align*}
&C(P)\|\nabla w \|_{L^2}^{a(q+1)}\\
\leq &\tilde{C}(P,\epsilon)+\epsilon \| \nabla w\|_{L^2}^2\\
\leq &\tilde{C}(P,\epsilon)+\epsilon \int |\nabla w |^2 ~d\mathbf{x}\\
\leq &\tilde{C}(P,\epsilon)+C'\epsilon \int \frac{e^{ s}}{2}|\nabla w|^2 ~d\mathbf{x}.
\end{align*}
Therefore,
\begin{equation*}
E(w)\geq (1-C'\epsilon) \int \frac{e^{ s}}{2}|\nabla w|^2~d\mathbf{x}  -\tilde{C}(P,\epsilon).
\end{equation*}
Choose $\epsilon$ so small that $(1-C'\epsilon)>\frac{1}{2}$, the assertion is established.
\end{proof}

Let us define:
\begin{equation}\label{chap4: infimum}
I_P=\inf_{w\in W_P}\{E(w)| N(w)=P\}.
\end{equation} 
lemma \ref{chap4: bounded from below lemma} shows that $I_P>-\infty$. We can quickly find a few useful scaling inequalities on $I_P$.

\begin{lemma}\label{chap4: scaling inequality}
Suppose $q> 1$. Given $s$ and $f$, $I_P<0$ for $P$ sufficiently large. If $P'>P>0$, then $I_{P'}\leq\bigg(\frac{P'}{P}\bigg)^{q+1}I_P$.
\end{lemma}

\begin{proof}
Notice in \eqref{chap4: constraint} that $N(w)$ is linear in w. We have for $\theta >1$,
\begin{align*}
I_{\theta P}&=\inf\big\{ E(w) ~\big|~ N(w)=\theta P\big\}\\
&=\inf\big\{ E(\theta w) ~\big|~ N(w)= P\big\}\\
&=\inf\big\{ \int \frac{e^{s}}{2} \theta ^2 |\nabla w|^2-\frac{K}{q+1}\theta ^{q+1}w^{q+1}e^{- s} |N(w)=P  \}.
\end{align*}
Now observe that 
\begin{equation*}
\int w^{q+1} e^{- s} >0
\end{equation*}
and the term with the coefficient $\theta ^{q+1}$ will dominate as $\theta$ increases, we can conclude that $I_{\theta P}<0$ if $\theta$ is sufficiently large.

Following the same line of reasoning,
\begin{align*}
I_{P'}&=\inf\big\{E(w) ~\big|~ N(w)=P'  \big\}\\
&=\inf\bigg\{ E\bigg(\bigg(\frac{P'}{P}\bigg) w\bigg) ~\big|~ N(w)= P\bigg\}\\
&=\inf\bigg\{ \int \frac{e^{s}}{2} \bigg(\frac{P'}{P}\bigg) ^2 |\nabla w|^2-\frac{K}{q+1}\bigg(\frac{P'}{P}\bigg) ^{q+1}w^{q+1}e^{- s} ~\big|~N(w)=P  \bigg\}\\
&=\bigg(\frac{P'}{P}\bigg)^{q+1}\inf\bigg\{ \int \frac{e^{ s}}{2} \bigg(\frac{P'}{P}\bigg) ^{1-q} |\nabla w|^2-\frac{K}{q+1}w^{q+1}e^{- s} ~\big|~N(w)=P  \bigg\}\\
&\leq\bigg(\frac{P'}{P}\bigg)^{q+1}\inf \bigg\{ \int \frac{e^{ s}}{2}|\nabla w |^2  -\frac{K}{q+1}w^{q+1}e^{- s} ~\big|~ N(w)=P\bigg\}\\
&=\bigg(\frac{P'}{P}\bigg) ^{q+1}I_P.
\end{align*}
We get the inequality because $P'>P$ and $q > 1$.
\end{proof}

We are now ready to introduce a concentration compactness principle due to Lions \cite{lions1984concentration}. This is the starting point of the existence argument.

\begin{lemma}\label{chap4: Lions lemma}
Let $\{w_n\}$ be a sequence in $L^1(\mathbb{R}^3)$ such that $w_n\geq 0$ a.e. Suppose $w_n$'s are axisymmetric, and $\int_{\mathbb{R}^3}f w_n ~d\mathbf{x}=P$. Then there exists a subsequence $\{w_{n_k}\}$ such that one of the following is true: 

\begin{enumerate}
\item {
$\exists\ \{a_k\}\in \mathbb{R}$ such that $\forall \epsilon >0, \exists R>0, K_0>0$ such that $\forall k>K_0$
\begin{equation*}
P\geq \int_{a_k \mathbf{e}_3 +S_R}f w_{n_k}~d\mathbf{x}\geq P-\epsilon.
\end{equation*}
}
\item{
$\forall R>0$
\begin{equation*}
\lim_{k\to \infty}\sup_{ \mathbf{y}\in \mathbb{R}^3}\int_{\mathbf{y}+S_R}f w_{n_k}~d\mathbf{x}=0.
\end{equation*}
}
\item{
$\exists \lambda \in(0,P),\forall \epsilon>0, \exists R_0>0, \{a_k\}\in \mathbb{R}, \forall R>R_0, \exists k_0>0, \forall k>k_0$:
\begin{align*}
\int_{a_k \mathbf{e}_3 +S_R}fw_{n_k} ~d\mathbf{x} &>\lambda - \epsilon ,\\
\int_{a_k \mathbf{e}_3 +S_{2R}}f w_{n_k} ~d\mathbf{x} &< \lambda + \epsilon .\\
\end{align*}
}
\end{enumerate}
\end{lemma}

\begin{proof}
Denote $f w_n$ by $\rho_n$. Let $ Q_n(t)=\sup_{\mathbf{y}\in \mathbb{R}^3}\int_{\mathbf{y}+S_t}\rho_n ~d\mathbf{x}$.

$Q_n(t)$ is a sequence of nondecreasing, nonnegative, uniformly bounded functions on $\mathbb{R}^+$, and $\lim_{t\to +\infty}Q_n(t)=P$. By the Helly selection theorem, there exists a subsequence $Q_{n_k}(t)$, and a function $Q(t)$, such that $Q_{n_k}(t)\to Q(t)$ pointwise on $\mathbb{R}^+$. $Q(t)$ is hence non-decreasing and non-negative.

Let $\lambda=\lim_{n\to \infty}Q(t)\in [0,P]$.
\begin{enumerate}
\item{
If $\lambda =P$, then 
$\forall \epsilon>0$, $\exists R(\epsilon)>0$ such that $Q(R)>P-\frac{\epsilon}{2}$.

Since $ \lim_{k\to \infty}Q_{n_k}(R)=Q(R)$,
$\exists K_0(\epsilon)>0$, $\forall k>K_0(\epsilon)$: $Q_{n_k}(R)>P-\frac{\epsilon}{2}$.

Hence, $\exists \mathbf{y}_k(\epsilon)\in \mathbb{R}^3$ such that $\int_{\mathbf{y}_k(\epsilon)+S_R}\rho_{n_k}~d\mathbf{x}>P-\frac{\epsilon}{2}$.
Take $\mathbf{y}_k=\mathbf{y}_k\bigg(\frac{P}{2}\bigg)$. We claim that $ |\mathbf{y}_k(\epsilon)-\mathbf{y}_k|<R\bigg(\frac{P}{2}\bigg)+R(\epsilon)$ for $\epsilon$ small. If not,
\begin{align*}
\int_{\mathbb{R}^3}\rho_{n_k}~d\mathbf{x} &\geq \int_{\mathbf{y}_k+R(\frac{P}{2})}\rho_{n_k}~d\mathbf{x}+\int_{\mathbf{y}_k(\epsilon)+R(\epsilon)}\rho_{n_k}~d\mathbf{x} \\
&>P-\frac{P}{2}+P-\frac{\epsilon}{2} \\
&=\frac{3P}{2}-\frac{\epsilon}{2}>P\quad \text{if } \epsilon \text{ is small.}
\end{align*} 
Take $ R'(\epsilon)=2R(\epsilon)+R\bigg(\frac{P}{2}\bigg)$. By the previous inequality, we have
\begin{equation*}
\mathbf{y}_k+R'(\epsilon)\supset \mathbf{y}_k(\epsilon)+R(\epsilon).
\end{equation*}
Therefore,
\begin{equation*}
\int_{\mathbf{y}_k+B_{R'(\epsilon)}}\rho _{n_k}~d\mathbf{x}>P-\frac{\epsilon}{2}.
\end{equation*}
}
Take $a_k=\mathbf{y}_k\cdot \mathbf{e}_3$, and let $r(\mathbf{y})$ be the distance of $\mathbf{y}$ to the $\mathbf{e}_3$ axis. There must exist an $r_0$ such that $r(\mathbf{y}_k)\leq r_0$. Otherwise the integral of $\rho_{n_k}$ on the torus obtained from revolving $\mathbf{y}_k+S_{R(\frac{P}{2})}$ around the $\mathbf{e}_3$ axis will give
\begin{equation*}
\int_{T_k}\rho_{n_k}~d\mathbf{x} \geq C\bigg(P-\frac{P}{2}\bigg)r(\mathbf{y}_k)
\end{equation*}  
for some constant $C$. The right hand side is bounded because the left hand side is.

Let $R''(\epsilon)=R'(\epsilon)+r_0$, then
\begin{equation*}
\int_{a_k \mathbf{e}_3+B_{R''(\epsilon)}}\rho_{n_k}~d\mathbf{x}> P-\epsilon.
\end{equation*}

\item{
If $\lambda =0$, then $ \lim_{R\to \infty}Q(R)=0$, which implies $Q(R)\equiv 0$. The result follows immediately.
}

\item{
If $\lambda\in (0,P)$, since $ \lim_{t\to \infty}Q(t)=\lambda, \lim_{k\to \infty} Q_{n_k}(t)=Q(t)$, we know:

$\forall \epsilon>0, \exists R(\epsilon)>0, K_0>0, \forall k>K_0, R\geq R(\epsilon)$:
\begin{equation*}
Q_n(R)=\sup_{\mathbf{y}\in \mathbb{R}^3}\int_{\mathbf{y}+S_R}\rho_{n_k}~d\mathbf{x}> \lambda -\epsilon.
\end{equation*}
Let $ f_k(\mathbf{y})=\int_{\mathbf{y}+B_{R(\epsilon)}} \rho_{n_k}~d\mathbf{x}$. It is easy to verify that $f_k(\mathbf{y})$ is a continuous function. Consider the set $ \big\{\mathbf{y}~\big|~f_k(\mathbf{y})\geq \lambda-\epsilon\big\}$. This set is nonempty because $ \sup_{\mathbf{y}\in \mathbb{R}^3}f_k(\mathbf{y})> \lambda-\epsilon$, is closed by the continuity of $f_k$, and is bounded because the contrary will indicate that $ \rho_{n_k}$ has infinite mass. Therefore, there exists $\mathbf{y}_k\in \mathbb{R}^3$ such that 
\begin{equation*}
f_k(\mathbf{y}_k)=\int_{\mathbf{y}_k+S_{R(\epsilon)}} \rho_{n_k}~d\mathbf{x} =\sup_{\mathbf{y}\in \mathbb{R}^3}\int_{\mathbf{y}+B_{R(\epsilon)}}\rho_{n_k}~d\mathbf{x} > \lambda -\epsilon.
\end{equation*}
Also for any $R\geq R(\epsilon)$, we have
\begin{equation*}
\int_{\mathbf{y}_k+S_R} \rho_{n_k}~d\mathbf{x} > \lambda-\epsilon.
\end{equation*}
For the same reason as in case 1, there must be an $r_0=r_0(\epsilon)$ such that $r(\mathbf{y}_k)\leq r_0$. Let $a_k=\mathbf{y}_k\cdot \mathbf{e}_3$, and $R_0=R(\epsilon)+r_0$, $\forall R>R_0, k>K_0$,
\begin{equation*}
\int_{a_k\mathbf{e}_3+S_R}\rho_{n_k}~d\mathbf{x}>\lambda-\epsilon.
\end{equation*}
}
On the other hand, because $ \lim_{k\to \infty}Q_{n_k}(2R)\leq \lambda$, there must be a $k_0>K_0$ such that $\forall k>k_0$:
\begin{equation*}
Q_{n_k}(2R)<\lambda+\epsilon,
\end{equation*}
which implies
\begin{equation*}
\int_{a_k\mathbf{e}_3+S_{2R}}\rho_{n_k}~d\mathbf{x} <\lambda +\epsilon.
\end{equation*}
\end{enumerate}
This concludes the proof of the lemma.
\end{proof}

Intuitively, lemma \ref{chap4: Lions lemma} says that if we have a sequence of densities with fixed total mass, then the densities will either concentrate in a ball of radius $R$, or vanish as $n$ goes to infinity, or split up into at least two parts (with masses roughly $\lambda$ and $M-\lambda$) that escape infinitely far from each other as $n$ goes to infity. Our analysis in the following will show that case 2 and case 3 cannot happen, provided that the scaling inequalities hold. On the other hand, case 1 will force the existence of a minimizer.

\begin{lemma}\label{chap4: vanishing mass lemma}
Let $1< q<3$. If $w_n$ is bounded in $L^1(\mathbb{R}^3)\cap H^1(\mathbb{R}^3)$, $w_n\geq 0$ a.e., and if 
\begin{center}
$\exists R>0$,  $\lim_{n \to \infty} \sup_{\mathbf{y}\in \mathbb{R}^3}\int_{\mathbf{y}+S_R}w_n ~d\mathbf{x}\to 0$,
\end{center}
Then $\int_{\mathbb{R}^3}w_n^{q+1}~d\mathbf{x}\to 0$.
\end{lemma}

\begin{proof}
Fix $  \alpha \in \bigg(\max \bigg\{\frac{3}{2},\frac{2(q+1)}{3}\bigg\},q+1\bigg)$, and let $ \beta=\frac{q+1}{\alpha}$. We get $ 1<\beta<\frac{3}{2}$. For any $w\in L^1(\mathbb{R}^3)\cap H^1(\mathbb{R}^3)$, by Sobolev embedding $W^{1,1} \subset L^{\beta}$,
\begin{align}\label{chap4: Sobolev inequality}
&\int_{\mathbf{y}+S_R}w^{q+1}~d\mathbf{x} \notag \\
=&\int_{\mathbf{y}+S_R}w^{\alpha \beta}~d\mathbf{x} \notag \\
\leq & C(R)\bigg( \int_{\mathbf{y}+S_R}(w^{\alpha}+\alpha w^{\alpha-1}|\nabla w|) ~d\mathbf{x} \bigg)^{\beta} \notag \\
\leq & C(R)\bigg( \int_{\mathbf{y}+S_R}w^{\alpha} ~d\mathbf{x} +\alpha \big[\int_{\mathbf{y}+S_R}w^{2(\alpha -1)} ~d\mathbf{x} \big]^{\frac{1}{2}}\big[\int_{\mathbf{y}+S_R}|\nabla w|^2 \big]^{\frac{1}{2}}\bigg)^{\beta} \notag \\
= & C(R)(\|w\|_{L^{\alpha}(\mathbf{y}+S_R)}^{\alpha} +\alpha \|\nabla w\|_{L^2(\mathbf{y}+S_R)}\cdot \| w\|_{L^{2\alpha -2}(\mathbf{y}+S_R)}^{\alpha -1})^{\beta}.
\end{align}
By the Gagliardo-Nirenberg inequality,
\begin{align*}
\| w \|_{L^{\alpha}(\mathbf{y}+S_R)} &\leq C(R)\|\nabla w\|_{L^2(\mathbf{y}+S_R)}^a\| w\|_{L^1(\mathbf{y}+S_R)}^{1-a}\\
\| w \|_{L^{2\alpha -2}(\mathbf{y}+S_R)} &\leq C(R)\|\nabla w\|_{L^2(\mathbf{y}+S_R)}^b \| w\|_{L^1(\mathbf{y}+S_R)}^{1-b},\\
\end{align*}
where
\begin{align*}
a &=\frac{1-\frac{1}{\alpha}}{1-\frac{1}{6}}\\
b &=\frac{1-\frac{1}{2\alpha-2}}{1-\frac{1}{6}}.\\
\end{align*}
One has $a, b\in (0,1)$ if $1<\alpha<6$, $1<2\alpha -2<6$, or $ \frac{3}{2}<\alpha <4$, which is guaranteed by the choice of $\alpha$. Hence given the hypotheses for $w_n$, we can easily get 
\begin{align*}
\| w_n \|_{L^{\alpha}(\mathbf{y}+S_R)} &\to 0,\\
\| w_n \|_{L^{2\alpha -2}(\mathbf{y}+S_R)} &\to 0,
\end{align*}
as $n\to \infty$.
By \eqref{chap4: Sobolev inequality}
\begin{align*}
& \int_{\mathbf{y}+S_R}w_n^{q+1}~d\mathbf{x} \\
\leq & C(R)\bigg( \int_{\mathbf{y}+S_R}(w_n^{\alpha}+\alpha w_n^{\alpha-1}|\nabla w_n|) ~d\mathbf{x} \bigg)^{\beta}  \\
:= & C(R) \epsilon_n ^{\beta} \\
\leq & C(R) \epsilon_n^{\beta-1} \int_{\mathbf{y}+S_R}(w_n^{\alpha}+\alpha w_n^{\alpha-1}|\nabla w_n|) ~d\mathbf{x} ,
\end{align*}
where
\begin{align*}
\epsilon_n &=\int_{\mathbf{y}+S_R}(w_n^{\alpha}+\alpha w_n^{\alpha-1}|\nabla w_n|) ~d\mathbf{x} \\
&\leq \| w_n \|_{L^{\alpha}(\mathbf{y}+S_R)}^{\alpha} +\alpha \|\nabla w_n\|_{L^2(\mathbf{y}+S_R)}\cdot \| w_n\|_{L^{2\alpha -2}(\mathbf{y}+S_R)}^{\alpha -1}\\
& \to 0.
\end{align*}
Cover $\mathbb{R}^3$ with these balls of radius $R$ in such a way that each point in $\mathbb{R}^3$ is contained in an overlap of at most $m$ balls. Then,
\begin{align}\label{bounds on R^3}
\int_{\mathbb{R}^3}w_n^{q+1}~d\mathbf{x} \leq C(R)m\epsilon_n^{\beta -1}\int_{\mathbb{R}^3}(w_n^{\alpha}+\alpha w_n^{\alpha -1}|\nabla w|)~d\mathbf{x}.
\end{align}
Just as in \eqref{chap4: Sobolev inequality}, we have
\begin{align*}
& \int_{\mathbb{R}^3}(w_n^{\alpha}+\alpha w_n^{\alpha -1}|\nabla w|)~d\mathbf{x} \\
\leq & \| w_n\|_{L^{\alpha}(\mathbb{R}^3)}^{\alpha} +\alpha \|\nabla w_n\|_{L^2(\mathbb{R}^3)}\cdot \| w_n\|_{L^{2\alpha -2}(\mathbb{R}^3)}^{\alpha -1}.
\end{align*}
Similarly by the Gagliardo-Nirenberg inequality,
\begin{align*}
\| w_n \|_{L^{\alpha}(\mathbb{R}^3)} &\leq C\|\nabla w_n\|_{L^2(\mathbb{R}^3)}^a\| w_n\|_{L^1(\mathbb{R}^3)}^{1-a}\\
\| w_n \|_{L^{2\alpha -2}(\mathbb{R}^3)} &\leq C\|\nabla w_n\|_{L^2(\mathbb{R}^3)}^b \| w_n\|_{L^1(\mathbb{R}^3)}^{1-b}.\\
\end{align*}
By the boundedness of $w_n$ in $L^1(\mathbb{R}^3)\cap H^1(\mathbb{R}^3)$, we conclude from \eqref{bounds on R^3} that 
\begin{equation*}
\int_{\mathbb{R}^3}w_n^{q+1}~d\mathbf{x} \to 0
\end{equation*}
as $n\to \infty$.
\end{proof}

\begin{corollary}
If $\{w_n\}$ is a minimizing sequence of $E$ in $W_P$, and if $I_P<0$, then case 2 in lemma \ref{chap4: Lions lemma} cannot happen.
\end{corollary}

\begin{proof}
If case 2 in lemma \ref{chap4: Lions lemma} happens, there will be a subsequence $\{w_{n_k}\}$ such that $\forall R>0$,
\begin{equation*}
\lim_{k\to \infty}\sup_{ \mathbf{y}\in \mathbb{R}^3}\int_{\mathbf{y}+S_R}f w_{n_k} ~d\mathbf{x}=0
\end{equation*}
Since $f\geq c >0$, this implies
\begin{equation*}
\lim_{k\to \infty}\sup_{ \mathbf{y}\in \mathbb{R}^3}\int_{\mathbf{y}+S_R}w_{n_k} ~d\mathbf{x}=0
\end{equation*}
By lemma \ref{chap4: vanishing mass lemma} and the boundedness of $s$ this implies $$\displaystyle \lim_{k\to \infty} \int_{\mathbb{R}^3}\frac{K}{q+1}e^{- s}w_{n_k}^{q+1} ~d\mathbf{x} =0,$$ which then implies $I_P\geq 0$. 
\end{proof}

For the purpose of eliminating the possibility of case 3 in lemma \ref{chap4: Lions lemma}, we need an elementary inequality.

\begin{lemma}\label{chap4: elementary inequality}
If $0\leq \lambda_1, \lambda_2 \leq 1$, $\lambda_1+\lambda_2=1$, $q>1$, then
\begin{equation*}
1-\lambda_1^{q+1}-\lambda_2^{q+1}\geq 2\lambda_1 \lambda_2.
\end{equation*}
\end{lemma}

\begin{proof}
Since $q>1$, $q+1>2$. Hence
\begin{align*}
1-\lambda_1^{q+1}-\lambda_2^{q+1} &\geq 1-\lambda_1^2-\lambda_2^2\\
& =(\lambda_1+\lambda_2)^2-\lambda_1^2-\lambda_2^2\\
&= 2\lambda_1 \lambda_2.
\end{align*}
\end{proof}

Now we are ready to eliminate case 3 in lemma \ref{chap4: Lions lemma}.
\begin{lemma}\label{chap4: splitting mass lemma}
Let $\{w_n\}$ be a minimizing sequence of $E$ in $W_P$. Suppose $I_P<0$, and $\forall P_2>P_1>0$, $I_{P_2}\leq\bigg(\frac{P_2}{P_1}\bigg)^{q+1}I_{P_1}$. Then case 3 in lemma \ref{chap4: Lions lemma} cannot happen.
\end{lemma}

\begin{proof}
Assume the contrary. Then there exists a subsequence $\{ w_{n_k}\}$ such that $\exists \lambda \in(0,P),\forall \epsilon>0, \exists R_0>0, a_k\in \mathbb{R}, \forall R>R_0, \exists k_0>0, \forall k>k_0$:
\begin{align}\label{chap4: mass bound on the ring}
\int_{a_k {\mathbf{e}_3} +S_R}f w_{n_k} ~d\mathbf{x} &>\lambda-\epsilon, \notag \\
\int_{a_k {\mathbf{e}_3} +S_{2R}}f w_{n_k} ~d\mathbf{x} &<\lambda+\epsilon.
\end{align}
Let $\varphi : \mathbb{R}^+ \to [0,1]$ be a smooth cut off function, such that
\begin{align*}
\varphi(t)&=1 \quad \text{when}\ \ |t|\leq 1,\\
\varphi(t)&=0 \quad \text{when}\ \ |t|\geq 2,\\
|\nabla \varphi (t)| &\leq 2 \quad \text{for all}\ \ t .
\end{align*}
Let us now define
\begin{align*}
\varphi_{k,1}(\mathbf{x}) &=\varphi\bigg(\frac{|\mathbf{x}-a_k{\mathbf{e}_3}|}{R}\bigg), \\
\varphi_{k,2}(\mathbf{x}) &=1-\varphi_{k,1}(\mathbf{x}),\\
w_{k,1}(\mathbf{x}) &=\varphi_{k,1}(\mathbf{x})w_{n_k}(\mathbf{x}),\\
w_{k,2}(\mathbf{x}) &=\varphi_{k,2}(\mathbf{x})w_{n_k}(\mathbf{x}),\\
P_{k,1}&=\int_{\mathbb{R}^3}f w_{k,1} ~d\mathbf{x},\\
P_{k,2}&=\int_{\mathbb{R}^3}f w_{k,2} ~d\mathbf{x}.\\
\end{align*}
Obviously $ w_{k,1}\in W_{P_{k,1}}, w_{k,2}\in W_{P_{k,2}}, |\nabla \varphi_{k,1}|\leq \frac{2}{R}, |\nabla \varphi_{k,2}|\leq \frac{2}{R}$, also $P= P_{k,1}+P_{k,2}$. We now estimate
\begin{align*}
E(w_{n_k}) &=\int_{\mathbb{R}^3}\bigg( \frac{e^{ s}}{2}|\nabla w_{n_k} |^2 -\frac{K}{q+1}w_{n_k}^{q+1}e^{- s} \bigg) ~d\mathbf{x}\\
 &= \int_{\mathbb{R}^3}\bigg( \frac{e^{ s}}{2}|\nabla w_{k,1} +\nabla w_{k,2}|^2 -\frac{K}{q+1}(w_{n_k}^{q+1}-w_{k,1}^{q+1}-w_{k,2}^{q+1})e^{- s}\\
 &\quad \quad -\frac{K}{q+1}(w_{k,1}^{q+1}+w_{k,2}^{q+1})e^{- s} \bigg) ~d\mathbf{x}\\
 &=\int_{\mathbb{R}^3}\bigg( \frac{e^{ s}}{2}|\nabla w_{k,1}|^2 -\frac{K}{q+1}w_{k,1}^{q+1}e^{- s} \bigg)~d\mathbf{x}\\
 &\quad +\int_{\mathbb{R}^3}\bigg( \frac{e^{ s}}{2}|\nabla w_{k,2}|^2 -\frac{K}{q+1}w_{k,2}^{q+1}e^{- s} \bigg)~d\mathbf{x}\\
 &\quad +\int_{\mathbb{R}^3}\bigg (e^{ s}\nabla w_{k,1}\cdot \nabla w_{k,2}-\frac{K}{q+1}w_{n_k}^{q+1}(1-\varphi_{k,1}^{q+1}-\varphi_{k,2}^{q+1})e^{- s}\bigg) ~d\mathbf{x}\\
 &\geq I_{P_{k,1}}+I_{P_{k,2}}+\int_{\mathbb{R}^3}\bigg (e^{ s}\nabla w_{k,1}\cdot \nabla w_{k,2}-\frac{K}{q+1}w_{n_k}^{q+1}(1-\varphi_{k,1}^{q+1}-\varphi_{k,2}^{q+1})e^{- s}\bigg) ~d\mathbf{x}\\
 &\geq \bigg[ \bigg( \frac{P_{k,1}}{P}\bigg)^{q+1}+\bigg( \frac{P_{k,2}}{P}\bigg)^{q+1}\bigg]I_P \\
 &\quad +\int_{\mathbb{R}^3}\bigg (e^{ s}\nabla w_{k,1}\cdot \nabla w_{k,2}-\frac{K}{q+1}w_{n_k}^{q+1}(1-\varphi_{k,1}^{q+1}-\varphi_{k,2}^{q+1})e^{- s}\bigg) ~d\mathbf{x}.
\end{align*}
The last inequality follows from the hypothesis in the lemma. If we denote
\begin{equation*}
Re=\int_{\mathbb{R}^3}\bigg (e^{ s}\nabla w_{k,1}\cdot \nabla w_{k,2}-\frac{K}{q+1}w_{n_k}^{q+1}(1-\varphi_{k,1}^{q+1}-\varphi_{k,2}^{q+1})e^{- s}\bigg) ~d\mathbf{x},
\end{equation*}
then the above estimate gives us
\begin{align*}
I_P-E(w_{n_k}) &\leq \bigg[ 1- \bigg( \frac{P_{k,1}}{P}\bigg)^{q+1}-\bigg( \frac{P_{k,2}}{P}\bigg)^{q+1}\bigg]I_P-Re.
\end{align*}
Since $P=P_{k,1}+P_{k,2}$, and $I_P<0$, by lemma \ref{chap4: elementary inequality}, we get
\begin{align*}
I_P-E(w_{n_k}) &\leq \bigg[ 1- \bigg( \frac{P_{k,1}}{P}\bigg)^{q+1}-\bigg( \frac{P_{k,2}}{P}\bigg)^{q+1}\bigg]I_P-Re\\
&\leq 2\frac{P_{k,1}P_{k,2}}{P^2}I_P-Re,\\
\end{align*}
or
\begin{equation}\label{chap4: bounds on P_k,2}
-\frac{2}{P^2}I_P P_{k,1} P_{k,2} \leq E(w_{n_k})-I_P-Re.
\end{equation}
Let us now estimate $Re$:
\begin{align*}
-Re &=-\int_{\mathbb{R}^3}\bigg (e^{ s}\nabla w_{k,1}\cdot \nabla w_{k,2}-\frac{K}{q+1}w_{n_k}^{q+1}(1-\varphi_{k,1}^{q+1}-\varphi_{k,2}^{q+1})e^{- s}\bigg) ~d\mathbf{x}.
\end{align*}
By the definition of $\varphi_{k,1}$ and $\varphi_{k,2}$, we know  $1-\varphi_{k,1}^{q+1}-\varphi_{k,2}^{q+1} \in [0,1]$, and is nonzero only when $R\leq |\mathbf{x}-a_k{\mathbf{e}_3}| \leq 2R$. Therefore
\begin{align*}
-Re &\leq -\int_{\mathbb{R}^3}e^{ s}\nabla w_{k,1}\cdot \nabla w_{k,2} ~d\mathbf{x} +C(q,K,s)\int_{R\leq |\mathbf{x}-a_k{\mathbf{e}_3}| \leq 2R}w_{n_k}^{q+1}~d\mathbf{x}\\
&= L_1 + L_2.
\end{align*}
We estimate $L_1$ and $L_2$ separately.
\begin{align*}
L_1 &= -\int_{\mathbb{R}^3}e^{ s}\nabla w_{k,1}\cdot \nabla w_{k,2} ~d\mathbf{x}\\
&= -\int_{\mathbb{R}^3}e^{ s}\nabla (w_{n_k}\varphi_{k,1})\cdot \nabla (w_{n_k}\varphi_{k,2}) ~d\mathbf{x}\\
&= -\int_{\mathbb{R}^3}e^{ s}\nabla \varphi_{k,1}\cdot \nabla \varphi_{k,2}|w_{n_k}|^2~d\mathbf{x} -\int_{\mathbb{R}^3}e^{ s}w_{n_k}\varphi_{k,2}\nabla \varphi_{k,1}\cdot \nabla w_{n_k} ~d\mathbf{x}\\
&\quad -\int_{\mathbb{R}^3}e^{ s}w_{n_k}\varphi_{k,1}\nabla \varphi_{k,2}\cdot \nabla w_{n_k} ~d\mathbf{x} -\int_{\mathbb{R}^3}e^{ s}\varphi_{k,1} \varphi_{k,2}|\nabla w_{n_k}|^2 ~d\mathbf{x}\\
&\leq -\int_{\mathbb{R}^3}e^{ s}\nabla \varphi_{k,1}\cdot \nabla \varphi_{k,2}|w_{n_k}|^2~d\mathbf{x} -\int_{\mathbb{R}^3}e^{ s}w_{n_k}\varphi_{k,2}\nabla \varphi_{k,1}\cdot \nabla w_{n_k} ~d\mathbf{x}\\
&\quad -\int_{\mathbb{R}^3}e^{ s}w_{n_k}\varphi_{k,1}\nabla \varphi_{k,2}\cdot \nabla w_{n_k} ~d\mathbf{x} \\
&\leq \frac{C(s)}{R}.
\end{align*}
The last inequality follows from $ |\nabla \varphi_{k,1}|\leq \frac{2}{R}, |\nabla \varphi_{k,2}|\leq \frac{2}{R}$ and that $\{w_{n_k}\}$ is bounded in $H^1(\mathbb{R}^3)$. On the other hand, by the Gagliardo-Nirenberg inequality,
\begin{align*}
L_2 &\leq C(q,K,s)\| w_{n_k}\| _{L^{q+1}(R\leq |\mathbf{x}-a_k{\mathbf{e}_3}|\leq 2R)}^{q+1}\\
&\leq C(q,K ,s)\| \nabla w_{n_k}\| _{L^2(R\leq |\mathbf{x}-a_k{\mathbf{e}_3}|\leq 2R)}^{a(q+1)} \| w_{n_k}\| _{L^1(R\leq |\mathbf{x}-a_k{\mathbf{e}_3}|\leq 2R)}^{(1-a)(q+1)}\\
&\leq C(q,K ,s)[(\lambda +\epsilon)-(\lambda -\epsilon)]^{(1-a)(q+1)}
\end{align*}
The constant $C(q,K ,s)$ is enlarged in different lines. The last inequality above follows from \eqref{chap4: mass bound on the ring}, and the fact that $\{w_{n_k}\}$ is bounded in $H^1(\mathbb{R}^3)$.\\
In summary, we have 
\begin{equation*}
-Re\leq \frac{C(s)}{R}+ C(q,K ,s)(2\epsilon)^{(1-a)(q+1)}.
\end{equation*}
From the range of $q$, we deduce that $a\in (0,1)$. Choose $R>R_0$ so big that 
\begin{equation*}
-Re\leq C(q,K ,s)\epsilon^{(1-a)(q+1)}.
\end{equation*}
By the definition of $w_{k,1}$, we have
\begin{align*}
P_{k,1}\geq \int_{|\mathbf{x}-a_k{\mathbf{e}_3}|\leq R}f w_{n_k}~d\mathbf{x}>\lambda -\epsilon.
\end{align*}
By \eqref{chap4: bounds on P_k,2}, and the estimates on $Re$, we have
\begin{equation*}
P_{k,2}\leq C(P,I_P,\lambda, q,K ,s)(\epsilon+\epsilon^{(1-a)(q+1)}).
\end{equation*}
However,
\begin{align*}
P_{k,2} &=\int_{\mathbb{R}^3}f w_{k,2} ~d\mathbf{x}\\
&\geq \int_{\mathbb{R}^3\setminus a_k{\mathbf{e}_3}+S_{2R}}f w_{n_k} ~d\mathbf{x}.
\end{align*}
Hence,
\begin{equation*}
\int_{\mathbb{R}^3\setminus a_k{\mathbf{e}_3}+S_{2R}}f w_{n_k} ~d\mathbf{x} \leq C(P,I_P,\lambda, q,K ,s)(\epsilon+\epsilon^{(1-a)(q+1)}).
\end{equation*}
On the other hand,
\begin{equation*}
\int_{a_k{\mathbf{e}_3}+S_{2R}}f w_{n_k} ~d\mathbf{x}<\lambda +\epsilon.
\end{equation*}
This implies
\begin{align*}
P &=\int_{\mathbb{R}^3}f w_{n_k} ~d\mathbf{x}\\
&<\lambda +\epsilon +C(P,I_P,\lambda, q,K ,s)(\epsilon+\epsilon^{(1-a)(q+1)}).
\end{align*}
If we have initially chosen $\epsilon$ so small that 
\begin{equation*}
\lambda +\epsilon +C(P,I_P,\lambda, q,K ,s)(\epsilon+\epsilon^{(1-a)(q+1)})<P.
\end{equation*}
a contradiction will be obtained.
\end{proof}

With the preparation above, we are ready to prove the existence of a minimizer.

\begin{proof}[Proof of proposition \ref{chap4: thm: existence of minimizer, infinite}]
By lemma \ref{chap4: scaling inequality}, the scaling inequalities are true in this $q$ range, therefore lemma \ref{chap4: Lions lemma}, lemma \ref{chap4: vanishing mass lemma} and lemma \ref{chap4: splitting mass lemma} apply. For any minimizing sequence $\{w_n\}$, there exists a subsequence $\{w_{n_k}\}$ such that case 1 in lemma \ref{chap4: Lions lemma} is true. Without loss of generality, we assume $\{w_n\}$ is already shifted, and satisfies:  $\forall \epsilon >0$, $\exists R>0, n_0>0$, $\forall n>n_0$:
\begin{equation*}
P\geq \int_{S_R}f w_n ~d\mathbf{x}\geq P-\epsilon.
\end{equation*}
By lemma \ref{chap4: bounded from below lemma}, $\{w_n\}$ is bounded in $H^1(\mathbb{R}^3)$. The Banach-Alaoglu theorem implies that there exists a subsequence of $\{w_n\}$ which converges weakly in $H^1(\mathbb{R}^3)$ to $\tilde{w}$. Without loss of generality, we call this subsequence $\{w_n\}$ again. We claim that $\tilde{w}$ is a minimizer of $E(w)$ in $W_M$.

Let us first show $\tilde{w}\in W_M$. Obviously $\tilde{w}\in H^1(\mathbb{R}^3)$. Notice for any $R>0$, we have $w_n\rightharpoonup w$ weakly in $H^1(S_R)$. By the Rellich-Kondrachov theorem, $H^1(S_R)$ is compactly embedded in $L^p(S_R)$ for $1\leq p<6$. This implies $\forall R>0$, $w_n\to \tilde{w}$ in $L^q(S_R)$ for $1\leq q<6$. The conditions $w\geq 0$ a.e. and $w$ axisymmetric are now easily established if we integrate the $w_n$'s against positive smooth test functions with compact supports and take the limit. Let us now show $N(\tilde{w})=P$. For that we observe $\forall \epsilon>0$, $\exists R>0, n_0>0$, $\forall n>n_0$:
\begin{equation*}
\int_{S_R}f w_n ~d\mathbf{x}\geq P-\epsilon.
\end{equation*}
Since $w_n\to \tilde{w}$ in $L^q(S_R)$ for all $R>0$, and $f$ is locally bounded, we have 
\begin{equation*}
\int_{S_R}f \tilde{w} ~d\mathbf{x}\geq P-\epsilon.
\end{equation*}
Therefore for any $\epsilon>0$
\begin{equation}\label{chap4: lower mass bound}
\int_{\mathbb{R}^3}f \tilde{w} ~d\mathbf{x} \geq P-\epsilon.
\end{equation}
On the other hand, for any $R>0$, 
\begin{equation*}
P\geq \int_{S_R}f w_n ~d\mathbf{x},
\end{equation*}
which implies
\begin{equation*}
P\geq \int_{S_R}f\tilde{w} ~d\mathbf{x},
\end{equation*}
which implies
\begin{equation}\label{chap4: upper mass bound}
P\geq \int_{\mathbb{R}^3}f \tilde{w} ~d\mathbf{x}.
\end{equation}
Combine \eqref{chap4: lower mass bound} and \eqref{chap4: upper mass bound}, we get
\begin{equation*}
\int_{\mathbb{R}^3}f \tilde{w} ~d\mathbf{x} = P.
\end{equation*}
This also shows $\tilde{w}\in L^1(\mathbb{R}^3)$. We have shown $\tilde{w}\in W_M$, it remains to establish the weak lower-semicontinuity of $E$. The first term in $E$ can be treated by the standard method. Let us observe that  
\begin{equation*}
F_c= \bigg\{w ~\bigg|~ \int_{\mathbb{R}^3}\frac{e^{ s}}{2}|\nabla w|^2 ~d\mathbf{x}\leq c \bigg\}
\end{equation*}
is a convex norm closed set in $H^1(\mathbb{R}^3)$, therefore is weakly closed. 

For the second term
\begin{equation*}
-\int_{\mathbb{R}^3}\frac{K}{q+1}w^{q+1}e^{- s}~d\mathbf{x}
\end{equation*}
we recall, $\forall \epsilon>0$, $\exists R>0, n_0>0$, $\forall n, n'>n_0$:
\begin{align*}
\int_{\mathbb{R}^3\setminus S_R}f w_n ~d\mathbf{x} &\leq \epsilon\\
\int_{\mathbb{R}^3\setminus S_R}f w_{n'} ~d\mathbf{x} &\leq \epsilon\\
\|w_n-w_{n'}\|_{L^{q+1}(S_R)}&< \epsilon.
\end{align*}
Therefore,
\begin{align*}
\|w_n-w_{n'}\|_{L^{q+1}(\mathbb{R}^3)} &\leq \|w_n-w_{n'}\|_{L^{q+1}(S_R)}+\|w_n-w_{n'}\|_{L^{q+1}(\mathbb{R}^3\setminus S_R)}\\
&< \epsilon +C\big(\sup_n\|w_n\|_{H^1(\mathbb{R}^3)}\big)^a\big(\|w_n\|_{L^1(\mathbb{R}^3\setminus S_R)}+\|w_{n'}\|_{L^1(\mathbb{R}^3\setminus S_R)}\big)^{1-a}\\
&\leq \epsilon +C'\epsilon .
\end{align*}
The second inequality above follows from the Gagliardo-Nirenberg inequality. Hence $\{w_n\}$ converges in $L^{q+1}(\mathbb{R}^3)$. But $w_n\to \tilde{w}$ in $L^{q+1}(S_R)$ for any $R>0$. This implies $w_n\to \tilde{w}$ in $L^{q+1}(\mathbb{R}^3)$. Therefore,
\begin{equation*}
\lim_{n\to \infty}\int_{\mathbb{R}^3}w_n^{q+1}e^{- s}~d\mathbf{x}=\int_{\mathbb{R}^3}\tilde{w}^{q+1}e^{- s}~d\mathbf{x}.
\end{equation*}
Combine the two terms in $E$. We have
\begin{equation*}
\liminf_{n\to \infty}E(w_n)\geq E(\tilde{w}).
\end{equation*}
This shows that $\tilde{w}$ is a minimizer.
\end{proof}

One can establish a similar proposition for functions restricted to a finite ball. Since one has compact Sobolev embedding theorems on bounded balls, the corresponding proof will be a lot easier. In particular, if we let
\begin{equation}\label{chap4: energy functional, finite}
E(w)=\int_{S_R}\bigg(\frac{e^{s}}{2}|\nabla w|^2-\frac{K}{q+1}w^{q+1}e^{- s}\bigg) ~d\mathbf{x},
\end{equation}
\begin{equation}\label{chap4: constraint, finite}
N(w)=\int_{S_R}f w ~d\mathbf{x},
\end{equation}
and let $W_P$ be
\begin{align*}
H_0^1(S_R)\cap L^1(S_R) \cap \big\{w: S_R\to \mathbb{R}, ~w\geq 0\ \text{a.e.}, w\ \text{is axisymmetric}, N(w)=P\big\},
\end{align*}
then proposition \ref{chap4: thm: existence of minimizer, infinite} with these newly defined $E$ and $W_P$ remains true. The proof for that is standard.

\section{Existence of Solution for Low Adiabatic Index}
\label{chap4: low index existence}

In this section, we give proofs to theorems \ref{chap4: solution 1<q<3 bounded} and \ref{chap4: solution 1<q<3 unbounded}. The argument is based on proposition \ref{chap4: thm: existence of minimizer, infinite}. We will only lay out the demonstration for the compactly supported case, i.e. theorem \ref{chap4: solution 1<q<3 bounded}. The whole space case is virtually identical.

We first study the Euler-Lagrange equation. Let $W$ be
\begin{align*}
H_0^1(S_R)\cap L^1(S_R) \cap \big\{w: S_R\to \mathbb{R}, ~w\geq 0\ \text{a.e.}, w\ \text{is axisymmetric}, N(w)<\infty \big\},
\end{align*}
one has

\begin{lemma}\label{chap4: lem: variational inequality}
$\exists \lambda\in \mathbb{R}, \forall u\in W$:
\begin{equation}\label{chap4: variational inequality}
\int_{S_R}\bigg(e^{ s}\nabla \tilde{w} \cdot \nabla(u-\tilde{w}) -Ke^{- s}\tilde{w}^q (u-\tilde{w})\bigg)~d\mathbf{x}\geq - \lambda \int_{S_R} f(u-\tilde{w})~d\mathbf{x}.
\end{equation}
\end{lemma}

\begin{proof}
Given $u\in W$, when $t>0$ is small enough,
\begin{equation*}
\tilde{w}+t\bigg[(u-\tilde{w})-\frac{N(u-\tilde{w})}{N(\tilde{w})}\tilde{w}\bigg] \in W_P,
\end{equation*}
therefore,
\begin{equation*}
\frac{d}{dt}E\bigg(\tilde{w}+t\bigg[(u-\tilde{w})-\frac{N(u-\tilde{w})}{N(\tilde{w})}\tilde{w}\bigg]\bigg)\bigg|_{t=0+}\geq 0.
\end{equation*}
Denote $ (u-\tilde{w})-\frac{N(u-\tilde{w})}{N(\tilde{w})}\tilde{w}$ by $\sigma$, we have
\begin{align*}
& \frac{E(\tilde{w}+t \sigma)-E(\tilde{w})}{t}\\
=& \int_{S_R}\bigg( e^{ s}\nabla \tilde{w} \cdot \nabla \sigma - \frac{K}{q+1}e^{- s}(q+1)(\tilde{w}+\theta \sigma)^q\sigma \bigg) ~d\mathbf{x} + O(t),
\end{align*}
where $\theta$ is between $0$ and $t$, and depends on $\mathbf{x}$. Take the limit as $t \to 0+$. By the dominated convergence theorem, we get
\begin{equation*}
\lim_{t \to 0+}\frac{E(\tilde{w}+t \sigma)-E(\tilde{w})}{t} = \int_{S_R}\bigg( e^{ s}\nabla \tilde{w} \cdot \nabla \sigma - Ke^{- s}\tilde{w}^q\sigma \bigg) ~d\mathbf{x}.
\end{equation*}
Denote this by $ E'_{\tilde{w}}(\sigma)$, we have
\begin{align*}
0 &\leq E'_{\tilde{w}}(\sigma)\\
& = E'_{\tilde{w}}(u-\tilde{w})-\frac{E'_{\tilde{w}}(\tilde{w})}{N(\tilde{w})}N(u-\tilde{w}).
\end{align*}
Letting $ - \lambda=\frac{E'_{\tilde{w}}(\tilde{w})}{N(\tilde{w})}$ completes the proof.
\end{proof}

\begin{lemma}\label{chap4: positivity of lambda}
If $I_P<0, q>1$, then $\lambda>0$.
\end{lemma}

\begin{proof}
Observe that $2\tilde{w} \in W$, therefore we may plug in $u=2\tilde{w}$ to find
\begin{align*}
-\lambda P & = - \lambda \int_{S_R} f(2\tilde{w}-\tilde{w})~d\mathbf{x}\\
& \leq \int_{S_R}(e^{ s} | \nabla  \tilde{w} | ^2 -Ke^{- s}\tilde{w}^{q+1})~d\mathbf{x}\\
& = \int_{S_R}(\frac{e^{ s}}{2}| \nabla \tilde{w} |^2 -\frac{K}{q+1}e^{- s}\tilde{w}^{q+1})~d\mathbf{x} + \int_{S_R}(\frac{e^{ s}}{2} |\nabla \tilde{w} |^2 -\frac{qK}{q+1}e^{- s}\tilde{w}^{q+1})~d\mathbf{x}\\
& \leq 2I_P\\
& < 0.
\end{align*}
\end{proof}

For any $\varphi \in C^{\infty}_0(S_R)$, $\varphi \geq 0$, let $S(\varphi)$ be
\begin{equation}
S(\varphi)(r,\theta,z)=\frac{1}{2\pi}\int_0^{2\pi} \varphi(r, \theta ,z)d\theta.
\end{equation}
Then $S(\varphi)$ is axisymmetric, $\tilde{w}+S (\varphi) \in W$, and 
\begin{align}\label{chap4: S calculation}
& \int_{S_R}\bigg(e^{ s}\nabla \tilde w\cdot \nabla \varphi -Ke^{- s}\tilde{w}^q\varphi + \lambda f\varphi \bigg)~d\mathbf{x} \notag \\
=  & \int_{S_R}\bigg(S(e^{ s})\nabla S(\tilde w)\cdot \nabla \varphi -KS(e^{- s}\tilde{w}^q)\varphi + \lambda S(f)\varphi \bigg) ~d\mathbf{x}\notag \\
= & \int_{S_R}\bigg(e^{ s}\nabla \tilde{w} \cdot S(\varphi) - Ke^{- s}\tilde{w}^q S(\varphi) +\lambda f S(\varphi) \bigg)~d\mathbf{x}\notag \\
\geq & 0.
\end{align}
One can pass from the second line to the third line by Fubini's theorem. The last line follows from \eqref{chap4: variational inequality}. Therefore,
\begin{equation}
-\nabla \cdot (e^{ s}\nabla \tilde{w})-Ke^{- s}\tilde{w}^q +\lambda f 
\end{equation}
is a positive distribution on $S_R$. By a theorem of Schwartz (see Schwartz \cite{schwartz1959theorie}), it must be a positive Borel measure:
\begin{equation}
-\nabla \cdot (e^{ s}\nabla \tilde{w})-Ke^{- s}\tilde{w}^q +\lambda f = d\mu.
\end{equation}
Let us write $\tilde{w}=w_1+w_2$, where $w_1 \in H^1_0(S_R)$ weakly solves
\begin{equation}
-\nabla \cdot (e^{ s} \nabla w_1)-Ke^{- s}\tilde{w}^q +\lambda f =0,
\end{equation}
and $w_2 \in H^1_0(S_R)$ weakly solves
\begin{equation}
-\nabla \cdot (e^{ s} \nabla w_2) = d\mu.
\end{equation}
From the range of $q$ and the fact that $\tilde{w} \in H^1_0(S_R)\subset L^6(S_R) $, we have $\tilde{w}^q \in L^2(S_R)$. By standard elliptic regularity theory, $w_1$ is continuous. We next show that $w_2$ is lower semicontinuous, following Lewy and Stampacchia \cite{lewy1969regularity}.

\begin{lemma}\label{chap4: Green representation of w_2}
Let $\tilde{S}$ be any ball contained in $S_R$. Let $G(\mathbf{x},\mathbf{y})$ be the Dirichlet Green's function of $\tilde{S}$ with respect to the operator $-\nabla \cdot (e^{ s}\nabla)$, i.e.
\begin{align*}
-\nabla_\mathbf{x}(e^{ s}\nabla_\mathbf{x} G(\mathbf{x},\mathbf{y})) & = \delta _\mathbf{y} \  \ \  \  \text{on } \tilde{S}\\
G(\mathbf{x},\mathbf{y}) & = 0 \      \  \       \ \text{on } \partial \tilde{S}
\end{align*}
then 
\begin{equation}\label{chap4: Green rep of w_2}
w_2(\mathbf{x})= \int_{\tilde{S}}G(\mathbf{x},\mathbf{y})d\mu(\mathbf{y})- \int_{\partial \tilde{S}}e^{ s(\mathbf{y})}\frac{\partial G(\mathbf{x},\mathbf{y})}{\partial n(\mathbf{y})}w_2(\mathbf{y})d\sigma(\mathbf{y})
\end{equation}
in $\tilde{S}$, where $\sigma$ is the standard surface measure on $\partial \tilde{S}$.
\end{lemma}

\begin{proof}
Pick any ball $S$ contained in $\tilde{S}$. $\forall \varphi \in C^{\infty}_0(S)$, $\varphi \geq 0$, $\exists u$ solving
\begin{align*}
-\nabla \cdot (e^{ s}\nabla u) &= \varphi \  \ \  \  \text{on } \tilde{S}\\
u &= 0 \  \ \  \  \text{on } \partial \tilde{S}
\end{align*}
It follows from the Green's theorem that
\begin{equation}
u(\mathbf{x})=\int_S G(\mathbf{x},\mathbf{y})\varphi (\mathbf{y}) ~d\mathbf{y}.
\end{equation}
Now
\begin{align*}
& \int_S w_2(\mathbf{x}) \varphi (\mathbf{x}) ~d\mathbf{x}\\
= & \int_{\tilde{S}}w_2(\mathbf{x}) \varphi (\mathbf{x}) ~d\mathbf{x}\\
= & -\int_{\tilde{S}}w_2 \nabla \cdot (e^{ s}\nabla u)~d\mathbf{x}\\
= & \int_{\tilde{S}}e^{ s}\nabla w_2\cdot \nabla u ~d\mathbf{x} - \int_{\partial \tilde{S}}e^{ s}w_2\frac{\partial u}{\partial n}d\sigma\\
= & \int_{\tilde{S}} u d\mu - \int_{\partial \tilde{S}}e^{ s}w_2 \frac{\partial u}{\partial n} d\sigma\\
= & \int_{\tilde{S}} \bigg( \int_S G(\mathbf{x},\mathbf{y}) \varphi(\mathbf{y}) ~d\mathbf{y} \bigg) d\mu(\mathbf{x}) - \int_{\partial \tilde{S}} e^{ s(\mathbf{x})}w_2(\mathbf{x})\bigg( \int_S\frac{\partial G(\mathbf{x},\mathbf{y})}{\partial n(\mathbf{x})}\varphi(\mathbf{y}) ~d\mathbf{y} \bigg) d\sigma(\mathbf{x})\\
=& \int_S \bigg( \int_{\tilde{S}}G(\mathbf{x},\mathbf{y})d\mu(\mathbf{x}) \bigg)\varphi(\mathbf{y})~d\mathbf{y} - \int_S \bigg(\int_{\tilde{S}}e^{ s(\mathbf{x})}w_2(\mathbf{x})\frac{\partial G(\mathbf{x},\mathbf{y})}{\partial n(\mathbf{x})} d\sigma(\mathbf{x})\bigg) \varphi(\mathbf{y}) ~d\mathbf{y}.
\end{align*}
The last equality follows from Fubini's theorem and the fact that $G(\mathbf{x},\mathbf{y})>0$ when $\mathbf{x} \neq \mathbf{y}$.
\end{proof}

\begin{proof}[Proof of theorem \ref{chap4: solution 1<q<3 bounded}]
Without loss of generality, we can assume $\alpha =1$ in \eqref{chap4: eq: rescaled eq}. When $\mathbf{x}$ is in a compact subset of $\tilde{S}$, and $\mathbf{y}$ on $\partial \tilde{S}$, $ \frac{\partial G(\mathbf{x},\mathbf{y})}{\partial n(\mathbf{y})}$ is a smooth function in $\mathbf{x}$ and $\mathbf{y}$. Hence first term in \eqref{chap4: Green rep of w_2} is continuous in $\mathbf{x}$. Also notice that $G(\mathbf{x},\mathbf{y})$ is a pointwise limit of 
\begin{equation*}
G_a(\mathbf{x},\mathbf{y})=
\begin{cases}
G(\mathbf{x},\mathbf{y}) &\mbox{if } G(\mathbf{x},\mathbf{y})\leq a\\
a &\mbox{if } G(\mathbf{x},\mathbf{y})>a
\end{cases}
\end{equation*}
and that 
\begin{equation}\label{chap4: int of G_a}
\int_{\tilde{S}}G_a(\mathbf{x},\mathbf{y}) d\mu(\mathbf{y})
\end{equation}
is continuous in $\mathbf{x}$ on $\tilde{S}$. By the monotone convergence theorem, 
\begin{equation*}
\int_{\tilde{S}}G(\mathbf{x},\mathbf{y}) d\mu(\mathbf{y})
\end{equation*}
is an increasing pointwise limit of \eqref{chap4: int of G_a}, and hence is lower semicontinuous. We can now conclude that $w_2$, and $\tilde{w}$ also, are lower semicontinuous. This implies that the set $U_+ = \big\{ \mathbf{x}\in S_R ~\big|~ \tilde{w}(\mathbf{x})>0 \big\}$ is open. If $\varphi \in C^{\infty}_0(U_+)$, then $\tilde{w} + t S(\varphi) \in W$ for $|t|$ sufficiently small. A similar calculation as \eqref{chap4: S calculation} will show that
\begin{equation}
\int_{S_R}\bigg(e^{ s}\nabla \tilde w\cdot \nabla \varphi -Ke^{- s}\tilde{w}^q\varphi + \lambda f\varphi\bigg) ~d\mathbf{x} = 0.
\end{equation}
In other words, $\tilde{w}$ solves
\begin{equation}
\nabla \cdot (e^{ s} \nabla w) + Ke^{- s}w^q - \lambda f = 0
\end{equation}
weakly on $U_+$. Regularity of the solution follows from standard elliptic regularity.
\end{proof}

\section{Existence of Solution for Given Gas Density}
\label{chap4: given gas configuration}

We prove theorems \ref{chap4: thm: solve velocity field: smooth} and \ref{chap4: thm: solve velocity field holder rho} in this section. To avail ourselves in establishing regularity at $r=0$, let us prove the following lemma.

\begin{lemma}\label{chap4: regularity at zero}
Let $f : (-\epsilon, \epsilon) \times (-\epsilon, \epsilon) \to \mathbb{R}$ be such that $f(-r,z) = -f(r,z) $, and assume that $f \in C^k$, $k\geq 1$, then the function
\begin{equation}
g(r,z)=
\begin{cases}
\frac{f(r,z)}{r}   \quad & r\neq 0\\
f_r(0,z) \quad & r=0
\end{cases}
\end{equation}
is in $C^{k-1}$.
\end{lemma}

\begin{proof}
Obviously $f(0,z)=0$, hence for $r\neq 0$,
\begin{align}
& g(r,z) \notag \\
= & \frac{1}{r}(f(r,z)-f(0,z)) \notag \\
= & \frac{1}{r}\int_0^r f_s(s,z) ds \notag \\
= & \frac{1}{r}\int_0^1 f_s(rs,z) r ds \notag \\
= & \int_0^1 f_s(rs,z) ds
\end{align}
Apparently the same equation is true for $r=0$, and the assertion is clear from this formula.
\end{proof}

\begin{proof}[Proof of theorem \ref{chap4: thm: solve velocity field: smooth}]
Let us write \eqref{chap4: eq: non-isentropic vector rotating star} in cylindrical coordinates:
\begin{equation}\label{chap4: eq: Euler-Poisson eq component}
\begin{cases}
p_r  = \rho(B\rho)_r+\rho r \Omega^2 \\
p_z =\rho(B\rho)_z \\ 
\end{cases}
\end{equation}
From the definition of $B\rho$, we get
\begin{equation}
(B\rho)_z(\mathbf{x})=\int_D \frac{\rho_z(\mathbf{y})}{|\mathbf{x}-\mathbf{y}|}~d\mathbf{y}.
\end{equation}
Therefore $B\rho_z(r,-z)=-B\rho_z(r,z)$ and $B\rho_z(r,z)>0$ when $z<0$, by hypothesis 4 and the symmetry of $\rho$ and $D$. Let
\begin{equation}\label{chap4: sol: pressure}
p(r,z)=\int_{-\psi(r)}^z \rho(r,\xi)B\rho_{\xi}(r,\xi)d\xi.
\end{equation}
From now on we allow $r$ to take negative values by evenly extending all the relevant functions across $r=0$. It is easily seen that $p>0$ in $D$, $p=0$ on $\partial D$ and that $p$ satisfies the second equation in \eqref{chap4: eq: Euler-Poisson eq component}. Since $\rho\in C^{k}(\bar{D})$ and $\partial D$ is smooth, we have $B\rho\in C^{k+1}(\bar{D})$. It is not difficult to see that $p\in C^k(\bar{D})$. Differentiate \eqref{chap4: sol: pressure} under the integral sign, we get
\begin{equation}
p_r(r,z)=\int_{-\psi(r)}^z \big( \rho_r(r,\xi)B\rho_{\xi}(r,\xi)+\rho(r,\xi) B\rho_{r \xi}(r,\xi)\big)~d\xi.
\end{equation}
By the first equation in \eqref{chap4: eq: Euler-Poisson eq component}, when $r>0$, $\Omega^2$ has to have the form:
\begin{align}\label{chap4: sol: ang vel}
\Omega^2 &= \frac{1}{r\rho} (p_r - \rho B\rho_r) \notag \\
 & = \frac{1}{r\rho} \bigg(\int_{-\psi(r)}^z \big( \rho_r B\rho_{\xi}+\rho B\rho_{r \xi}\big)~d\xi - \rho B\rho_r \bigg)\notag \\
 &= \frac{1}{r\rho} \bigg(\int_{-\psi(r)}^z \big( \rho_r B\rho_{\xi} + \rho B\rho_{r \xi}\big)~d\xi - \int_{-\psi(r)}^z \big(\rho B\rho _r\big)_{\xi} ~d\xi \bigg) \notag \\ 
 &= \frac{1}{r\rho} \int_{-\psi(r)}^z \big( \rho_r B\rho_{\xi} - \rho_{\xi} B\rho_r \big)~d\xi.
\end{align}
\eqref{chap4: sol: ang vel} is non-negative on $D$ by hypothesis 3. Define $\Omega^2$ by \eqref{chap4: sol: ang vel}, when $r>0$, and if $D$ contains points at $r=0$, by
\begin{equation}
\Omega^2(0,z)= \frac{1}{\rho}(p_r - \rho B\rho_r)_r(0,z)
\end{equation} 
Notice $p_r - \rho (B\rho)_r$ is odd in $r$. By lemma \ref{chap4: regularity at zero}, $\rho \Omega^2 \in C^{k-2}(D)$, hence so is $\Omega^2$. Such $p$ and $\Omega^2$ obviously satisfy \eqref{chap4: eq: Euler-Poisson eq component}. It remains to show that $\Omega^2$ extends to a continuous function on $\bar{D}$. Let us consider the following three cases: 
\begin{enumerate}
\item{
Let $r_0$ be a nonzero radius such that $(r_0,-\psi(r_0))\in \partial D$ and $\psi(r_0)>0$. 
\begin{align}
& \lim_{(r,z) \to (r_0,-\psi(r_0))} \Omega^2(r,z) \notag \\
= & \lim_{(r,z) \to (r_0,-\psi(r_0))} \frac{1}{r\rho} \int_{-\psi(r)}^z \big( \rho_r B\rho_{\xi} - \rho_{\xi} B\rho_r \big)~d\xi \notag \\
= & \lim_{(r,z) \to (r_0,-\psi(r_0))} \frac{1}{r \rho_z} \big( \rho_r B\rho_z  - \rho_z B\rho_r \big) \notag \\
= & \frac{1}{r_0 \rho_z(r_0,-\psi(r_0))} \big( \rho_r B\rho_z - \rho_z B\rho_r \big)(r_0,-\psi(r_0)).
\end{align}
Here we have used the differential mean value theorem and hypothesis 4.
}
\item{
If $\partial D$ contains points at $r=0$, since $\partial D$ is smooth and symmetric about $z=0$ and $r=0 $, we have $\psi(0)>0$, $\psi'(0)=0$. Hence
\begin{align}
& \lim_{\substack{(r,z) \to (0,-\psi(0)) \\ r\neq 0}} \Omega^2(r,z) \notag \\
= & \lim_{\substack{(r,z) \to (0,-\psi(0)) \\ r\neq 0}} \frac{1}{r\rho} \int_{-\psi(r)}^z \big( \rho_r B\rho_{\xi} - \rho_{\xi} B\rho_r \big)~d\xi \notag \\
= & \lim_{\substack{(r,z) \to (0,-\psi(0)) \\ r\neq 0}} \frac{1}{r\rho_z}  \big( \rho_r B\rho_z - \rho_z B\rho_r \big) \notag \\
= & \frac{1}{\rho_z(0,-\psi(0))} \big( \rho_r B\rho_z- \rho_z B\rho_r \big)_r(0,-\psi(0)),
\end{align}
and
\begin{align}
& \lim_{z \to -\psi(0)} \Omega^2(0,z) \notag \\
= & \lim_{z \to -\psi(0)} \frac{1}{\rho}(p_r - \rho B\rho_r)_r(0,z) \notag \\
= & \lim_{z \to -\psi(0)} \frac{1}{\rho(0,z)}  \bigg( \int_{-\psi(0)}^z \big( \rho_rB\rho_{\xi} - \rho_{\xi} B\rho_r \big)_r d\xi \notag \\
& \qquad \qquad \qquad \quad  + \psi'(0) \big( \rho_rB\rho_z - \rho_z B\rho_r \big)(0,-\psi(0))  \bigg)\notag \\
= & \frac{1}{\rho_z(0,-\psi(0))} \big( \rho_r B\rho_z - \rho_z B\rho_r \big)_r(0,-\psi(0)).
\end{align}
}
\item{
Let $r_0$ be such that $\psi(r_0)=0$. When $(r,z)$ gets close to $(r_0,0)$ and $z\leq 0$, we observe by the differential mean value theorem that
\begin{align}
& \Omega^2(r,z) \notag \\
= &\frac{1}{r \rho} \int_{-\psi(r)}^z \big( \rho_r B\rho_{\xi} - \rho_{\xi} B\rho_r \big)~ d\xi \notag \\
= &\frac{1}{r \rho_z} \big( \rho_r B\rho_z - \rho_z B\rho_r \big)(r, z') \notag \\
= &\frac{1}{r\rho_{zz}} \big( \rho_r B\rho_z - \rho_z B\rho_r \big)_z(r,z''),
\end{align}
where $z'$ is between $-\psi(r)$ and $z$, and $z''$ is between $z'$ and $0$. Therefore
\begin{align}
\lim_{(r,z) \to (r_0,0)}\Omega^2(r,z) = \frac{1}{r\rho_{zz}} \big( \rho_r B\rho_z - \rho_z B\rho_r \big)_z(r_0,0).
\end{align}
We use hypothesis (a) to conclude that the limit is finite.
}
\end{enumerate}
\end{proof}

\begin{remark}
It is possible to establish higher regularity for $\Omega^2$ at the first two types of boundary points. However, at the third type of boundary points, $\psi'(r_0)=\infty$, in order to get higher regularity, we need conditions on how fast $\psi'(r)$ grows at around $r_0$, which we do not employ ourselves doing here.
\end{remark}

With relaxed regularity conditions at the boundary, the same computation works if further growth conditions are imposed on the derivatives of $\rho$ when close to the boundary. Let us now give

\begin{proof}[Proof of theorem \ref{chap4: thm: solve velocity field holder rho}]
As before, we define 
\begin{equation}
p(r,z)=\int_{-\psi(r)}^z \rho(r,\xi)  B\rho_{\xi}(r,\xi)~d\xi.
\end{equation}
It follows from hypothesis 1, 2, 4 that $p>0$ in $D$, $p=0$ on $\partial D$ and that $p$ satisfies the second equation in \eqref{chap4: eq: Euler-Poisson eq component}. Since $\rho\in C^{\beta}(\bar{D})$, we have $ B\rho \in C^{2,\beta}(\bar{D})$, hence $p \in C^0(\bar{D})$. Now let us calculate the $r$ partial derivative of $p$. In the following, let $ b=\max_{r\leq s \leq r+h}\big( -\psi(s) \big)$.
\begin{align}
& \frac{1}{h}\bigg( \int_{-\psi(r+h)}^z \rho  B\rho_{\xi}(r+h, \xi) ~d\xi -\int_{-\psi(r)}^ z \rho B\rho_\xi (r,\xi) ~d\xi \bigg)\notag \\
= & \frac{1}{h}\bigg( \int_{-\psi(r+h)}^b \rho B\rho_{\xi}~d\xi - \int _{-\psi(r)}^b \rho B\rho_{\xi}~d\xi + \int_b^ z \big( \rho B\rho_{\xi}(r+h, \xi)- \rho  B\rho_\xi (r,\xi) \big)~d\xi \bigg).
\end{align}
It is easily seen that the first two terms converge to $0$ as $h$ goes to $0$. Let us focus on the last term:
\begin{align}\label{chap4: cal r partial of p}
& \frac{1}{h} \int _b^z \big( \rho B\rho_{\xi}(r+h, \xi)- \rho  B\rho_\xi (r,\xi) \big)~d\xi \notag \\
= & \int _b^z \big( \rho B\rho_{\xi}  \big)_r (r',\xi) ~d\xi,
\end{align}
where $r'$ is between $r$ and $r+h$. We will use the dominated convergence theorem to compute the limit of \eqref{chap4: cal r partial of p}. For that purpose we need an estimate on $\chi_{(b,z)}\big( \rho B\rho_{\xi}  \big)_r (r',\xi)$. For the moment let us assume $z<0$. By hypothesis 5, there is a $C>0$ such that $|\rho_r|\leq C|\rho_{\xi}|$, $|\rho_{rr}|\leq C|\rho_{\xi}|$, $|\rho_{r\xi}|\leq C |\rho_{\xi}|$ for $\xi< z$. Also $\rho_{\xi}>0$. Therefore
\begin{align}\label{chap4: estimates of integrand for r partial of p}
& |\big( \rho B\rho_{\xi}  \big)_r (r',\xi)| \notag \\
\leq & C_1|\rho_r(r',\xi)|+C_2|\rho (r',\xi)| \notag \\
\leq & C_1 \bigg| \rho_r(r,\xi)+ \int _r^{r'}\rho_{ss}(s,\xi)ds \bigg| + C_2 \bigg|\rho(r,\xi)+\int _r^{r'}\rho_s(s,\xi)ds \bigg| \notag \\
\leq & (C_1+C_2)\bigg(|\rho_r(r,\xi)|+|\rho(r,\xi)|+C\int _{r-h_0}^{r+h_0}\rho_{\xi}(s,\xi)ds \bigg) \notag \\
\leq & \tilde{C}\bigg( \rho_{\xi}(r,\xi) + \rho(r,\xi) + \int _{r-h_0}^{r+h_0}\rho_{\xi}(s,\xi)ds \bigg),
\end{align}
for some fixed $h_0>h$. In the integral term, if $(s,\xi)$ lies outside $D$, then extend $\rho_{\xi}$ to be $0$. The fact that the integral of \eqref{chap4: estimates of integrand for r partial of p} is finite is manifested by the following:
\begin{equation}
\int _{-\psi(r)}^z \rho_{\xi}(r,\xi) ~d\xi = \rho(r,z)
\end{equation}
\begin{equation}
\int _{-\psi(r)}^z  \rho(r,\xi) ~d\xi < \infty
\end{equation}
\begin{align}
& \int _{-\psi(r)}^z \int _{r-h_0}^{r+h_0}\rho_{\xi}(s,\xi)~ds d\xi \notag \\
\leq & \int _{r-h_0}^{r+h_0} \int _{-\psi(s)}^z \rho_{\xi}(s,\xi) ~d\xi ds \notag \\
\leq & \int _{r-h_0}^{r+h_0} \rho(s,z)ds \notag \\
< & \infty.
\end{align}
Therefore, by the dominated convergence theorem,
\begin{equation}
p_r(r,z)=\int_{-\psi(r)}^z \big( \rho B\rho_\xi \big)_r ~d\xi.
\end{equation}
Now if $z \geq 0$, the integral in \eqref{chap4: cal r partial of p} can be broken into two pieces: one from $b$ to $z'$ and the other from $z'$ to $z$, for some $z'<0$. Notice that the second piece lies completely inside $D$, where $\rho$ is $C^2$, so the limit is the same as before. We have proved $p \in C^1(D)$. Now define $\Omega^2$ by
\begin{equation}\label{chap4: omega square}
 \frac{1}{r\rho} \int_{-\psi(r)}^z \big( \rho_r B\rho_{\xi} - \rho_{\xi} B\rho_r \big)~d\xi.
\end{equation}
when $r>0$, and if $D$ contains points at $r=0$, by
\begin{equation}
\Omega^2(0,z)= \frac{1}{\rho}\int_{-\psi(0)}^z \big( \rho_r B\rho_{\xi} - \rho_{\xi} B\rho_r \big)_r ~d\xi
\end{equation} 
The convergence of these integrals are guaranteed by hypothesis 5. It is easy to verify that such $p$ and $\Omega^2$ satisfy \eqref{chap4: eq: Euler-Poisson eq component}. Let us now show that $\Omega^2$ is continuous on $D$. Since $\partial D$ is smooth and convex at $(0,-\psi(0))$, $ -\psi(r)=\max_{0\leq s \leq r}\big( -\psi(s) \big)$ for $r$ small enough. Therefore,
\begin{align}
& \frac{1}{r}\int_{-\psi(r)}^z \big( \rho_r B\rho_{\xi} - \rho_{\xi} B\rho_r \big)~d\xi \notag \\
= & \int_{-\psi(r)}^z \big( \rho_r B\rho_{\xi} - \rho_{\xi} B\rho_r \big)_r(r',\xi)~d\xi,
\end{align}
where $r'$ is between $0$ and $r$. As before we assume $z<0$ and estimate the integrand,
\begin{align}\label{chap4: estimates of omega as r goes to zero}
& |\big( \rho_r B\rho_{\xi} - \rho_{\xi} B\rho_r \big)_r(r',\xi)| \notag \\
\leq & C_1(|\rho_{rr}(r',\xi)|+|\rho_{r\xi}(r',\xi)|+|\rho_r(r',\xi)|+|\rho_{\xi}(r',\xi)|) \notag \\
\leq & \tilde{C}\rho_{\xi}(r',\xi) \notag \\
\leq & \tilde{C}\bigg( \rho_{\xi}(0,\xi) + \int_0^{r'}\rho_{s\xi}(s,\xi) ds \bigg)\notag \\
\leq & \tilde{C}\bigg( \rho_{\xi}(0,\xi) + \int_0^{r_0}\rho_{\xi}(s,\xi) ds \bigg), 
\end{align}
where $r_0>r$ is small and fixed. As before, \eqref{chap4: estimates of omega as r goes to zero} has a finite $\xi$ integral. By the dominated convergence theorem, 
\begin{equation}\label{chap4: cont at r=0}
\lim_{\substack{(r,z)\rightarrow (0,z_0)\\r\neq 0, z_0<0}}\Omega^2(r,z)=\Omega^2(0,z).
\end{equation}
Again by splitting the integral \eqref{chap4: omega square} into boundary and interior parts, \eqref{chap4: cont at r=0} continues to be true when $z_0\geq 0$, and $\Omega^2(0,z)$ is evidently continuous in $z$. This establishes the continuity of $\Omega^2$ at $D\cap \{r=0\}$. The continuity of $\Omega^2$ away from the $z$ axis is obvious. It remains to show that $\Omega^2\in L^{\infty}(D)$. Let us consider the following three cases: 
\begin{enumerate}
\item{
Let $r_0$ be a nonzero radius such that $(r_0,-\psi(r_0))\in \partial D$ and $\psi(r_0)>0$. When $ z<-\frac{1}{2}\psi(r_0)$,
\begin{align}
& \Omega^2(r_0,z) \notag \\
= & \frac{1}{r_0\rho} \int_{-\psi(r_0)}^z \big( \rho_r B\rho_{\xi} - \rho_{\xi} B\rho_r \big)~d\xi \notag \\
= & \frac{1}{r_0 \rho_z} \big( \rho_r B\rho_z - \rho_z B\rho_r \big)(r_0,z') \notag \\
\leq & \frac{1}{r_0}(C| B\rho_z|+| B\rho_r|)\notag \\
\leq & \frac{\tilde{C}}{r_0},
\end{align}
where $C$ is given by hypothesis 5.
}
\item{
If $\partial D$ contains points at $r=0$, as $(r,z)$ gets close to $(0,-\psi(0))$, $r\neq 0$,
\begin{align}
& \Omega^2(r,z) \notag \\
= & \frac{1}{r\rho} \int_{-\psi(r)}^z \big( \rho_r B\rho_{xi} - \rho_{\xi} B\rho_r \big)~d\xi \notag \\
= & \frac{1}{r\rho_z}  \big( \rho_r B\rho_z - \rho_z B\rho_r \big)(r,z') \notag \\
= & \frac{\big( \rho_r B\rho_z - \rho_z B\rho_r \big)_r}{\rho_z + r \rho_{rz}}(r',z') \notag \\
\leq & \frac{\tilde{C}}{1+r' C}
\end{align}
where $z'$ is between $-\psi(r)$ and $z$, and $r'$ is between $0$ and $r$. In this process we have used the mean value theorem several times, the justification being that the convexity of $\partial D$ at $(0,-\psi(0))$ guarantees that all the relevant segments lie inside $D$. On the other hand if $r=0$,
\begin{align}
& \Omega^2(0,z) \notag \\
= & \frac{1}{\rho}\int_{-\psi(0)}^z \big( \rho_r B\rho_{\xi} - \rho_{\xi} B\rho_r \big)_r ~d\xi \notag \\
= & \frac{1}{\rho_z} \big( \rho_r B\rho_z - \rho_z B\rho_r \big)_r(0,z')\notag \\
\leq & \tilde{C}.
\end{align}
}
\item{
Let $r_0$ be such that $\psi(r_0)=0$. When $(r,z)$ gets close to $(r_0,0)$ and $z\leq 0$, 
\begin{align}
& \Omega^2(r,z) \notag \\
= &\frac{1}{r \rho} \int_{-\psi(r)}^z \big( \rho_r  B\rho_{\xi} - \rho_{\xi}  B\rho_r \big) ~d\xi \notag \\
= &\frac{1}{r \rho_z} \big( \rho_r  B\rho_z - \rho_z  B\rho_r \big)(r, z') \label{chap4: omega close to boundary on z=0 1} \\ 
= &\big( \frac{\rho_r  B\rho_z}{r \rho_z}  - \frac{ B\rho_r}{r} \big)(r, z') \label{chap4: omega close to boundary on z=0 2} \\ 
= &\frac{1}{r\rho_{zz}} \big( \rho_r  B\rho_z\big)_z(r,z'')- \frac{ B\rho_r}{r} (r, z')\label{chap4: omega close to boundary on z=0 3},
\end{align}
where $z'$ is between $-\psi(r)$ and $z$, and $z''$ is between $z'$ and $0$. If hypothesis (a) is satisfied, by \eqref{chap4: omega close to boundary on z=0 3}, $\Omega^2(r,z)$ is bounded. If hypothesis (a') is satisfied, by \eqref{chap4: omega close to boundary on z=0 1} and the fact that $\rho_z>0$, $ B\rho_z>0$ when $z<0$, 
\begin{align}
0\leq \Omega^2(r,z)\leq -\frac{1}{r} B\rho_r(r,z').
\end{align}
Therefore $\Omega^2(r,z)$ is bounded. If hypothesis (a'') is satisfied, by \eqref{chap4: omega close to boundary on z=0 2} and the fact that $| B\rho_z|<C|z|$, we again get the boundedness of $\Omega^2(r,z)$.
}
\end{enumerate}
\end{proof}

\section*{Acknowledgment}
This work is included in the Ph.D. dissertation of the author under the direction of Professor Joel Smoller at the University of Michigan.

\bibliographystyle{acm} 
\bibliography{biblio}   

\end{document}